\documentclass[preprint]{elsarticle}

\usepackage[dvipsnames]{xcolor}
\usepackage{lineno,hyperref}
\modulolinenumbers[5]

\usepackage{graphicx}
\usepackage{latexsym} 
\usepackage{verbatim}
\usepackage{amsmath}

\usepackage{amsthm}
\usepackage{amsfonts}
\usepackage{amssymb}
\usepackage{mathabx}
\usepackage{tikz}
\usetikzlibrary{positioning, shapes, arrows, snakes}
\usetikzlibrary{backgrounds}
\usepackage{tikz-cd}
\usepackage{mathpartir}
\usepackage{bussproofs}
\usepackage{enumitem}

\newcommand{\flld}[1]{\ld{\sqsupseteq}\ }
\newcommand{\frld}[1]{\ld{\sqsubseteq}\ }
\newcommand{\flrd}[1]{\rd{\sqsupseteq}\ }
\newcommand{\frrd}[1]{\rd{\sqsubseteq}\ }

\newcommand{\conjp}{\curlywedge}
\newcommand{\disjp}{\curlyvee}

\newcommand{\ld}[1]{\textcolor{MidnightBlue}{#1}}
\newcommand{\rid}[1]{\textcolor{RedOrange}{#1}}
\newcommand{\tid}[1]{\textcolor{OliveGreen}{#1}}
\newcommand{\rd}[1]{\textcolor{RedOrange}{#1'}}
\newcommand{\mt}[2]{
	\ifthenelse{\equal{#2}{}}{\langle #1,\rd{#1}\rangle}{\langle #1_{#2},\rd{#1_{#2}}\rangle}
}
\newcommand{\mtsq}[2]{
	\ifthenelse{\equal{#2}{}}{[#1,\rd{#1}]}{[#1_{#2},\rd{#1_{#2}}]}
}

\theoremstyle{plain}
\newtheorem{otheorem}{Theorem}

\newtheorem{oproposition}[otheorem]{Proposition}

\newtheorem{definition}{Definition}

\newtheorem{theorem}[otheorem]{Theorem}
\newtheorem{lemma}[otheorem]{Lemma}
\newtheorem{proposition}[otheorem]{Proposition}
\newtheorem{corollary}[otheorem]{Corollary}

\makeatletter
\def\ps@pprintTitle{%
 \let\@oddhead\@empty
 \let\@evenhead\@empty
 \def\@oddfoot{}%
 \let\@evenfoot\@oddfoot}
\makeatother

\begin{document}

\begin{frontmatter}

\title{Secure Information Flow Connections} 

\author{Chandrika Bhardwaj\fnref{cbfootnote}}
\ead{chandrika.bhardwaj@gmail.com}
\author{Sanjiva Prasad\fnref{spfootnote}}
\ead{sanjiva@cse.iitd.ac.in}
\address{Indian Institute of Technology Delhi, INDIA}

%
%

\begin{abstract}

Denning's \textit{lattice model} provided secure information flow analyses with an intuitive mathematical foundation:  the lattice ordering determines permitted flows.
We examine how this framework may be extended to support the flow of information between autonomous organisations, each employing possibly quite different security lattices and information flow policies.  
We propose a connection framework that permits different organisations to exchange information while maintaining both security of information flow as well as their autonomy in formulating and maintaining security policies.
Our prescriptive framework is based on the rigorous mathematical framework of \textit{Lagois connections} proposed by Melton, together with a simple operational model for transferring object data between domains.
The merit of this formulation is that it is simple, minimal, adaptable and intuitive.
We show that our framework is semantically sound, by proving that the connections proposed preserve standard correctness notions such as non-interference.
We then illustrate how Lagois theory also provides a robust framework and methodology for negotiating and maintaining secure agreements on information flow between autonomous organisations, even when either or both organisations change their security lattices.
Composition and decomposition properties indicate support for a modular approach to secure flow frameworks in complex organisations.
We next show that this framework extends naturally and conservatively to the Decentralised Labels Model of Myers \textit{et al.} --- a Lagois connection between the hierarchies of principals in two organisations naturally induces a Lagois connection between the corresponding security label lattices, thus extending the security guarantees ensured by the decentralised model to encompass bidirectional inter-organisational flows. 


\end{abstract}

\end{frontmatter}


\section{Introduction}

Denning's seminal work \cite{Denning76} established \textit{complete lattices} as the mathematical basis for a variety of analyses regarding \textit{secure information flow} (SIF), \textit{i.e.}, showing only authorised flows of information are possible.
An information flow model (IFM) 
$\langle N, P, SC, \sqcup, \sqsubseteq \rangle$ consists of storage objects $N$, which are assigned \textit{security classes} drawn from a (finite) complete lattice $SC$. 
The partial ordering $\sqsubseteq$ represents \textit{permitted flows} between classes. 
Reflexivity and transitivity capture intuitive aspects of information flow; antisymmetry helps avoid redundancies in the framework. 
The join operation $\sqcup$ succinctly captures the combination of information belonging to different security classes in arithmetic, logical and computational operations. 
$P$ is a set of processes, which are assigned security classes as ``clearances''.

The ensuing decades have seen a plethora of static and dynamic analysis techniques using that framework for programming languages \cite{sabelfeld2003language,myers1999jflow,Pottier2003-FlowCaml,liu2017fabric,roy2009laminar,Lourenco2015-ug}, operating systems \cite{Krohn2007-aa,zeldovich2006-osdi,cheng2012aeolus,efstathopoulos2005asbestos,roy2009laminar}, databases \cite{schultz2013ifdb}, and hardware architectures \cite{ferraiuolo2018hyperflow, zhang2015secVerilog-asplos}, etc. 
However, the question on how information can flow securely between independent organisations each with possible quite different policies has not been adequately addressed.
An answer needs to indicate how security classes from one lattice are mapped to those in another.
In doing so, we wish to abjure \textit{ad hoc} approaches to reclassifying information. 

We revisit and expand on our previous work \cite{BhardwajP2019}, where we proposed a simple and versatile mathematical framework involving \textit{monotone} increasing functions between lattices, which guaranteed secure and modular inter-organisational flows of information. 
Our work is based on the observation that large information systems are not monolithic, and are often constructed by autonomous organisations, each with its own security lattice, policies and mechanisms, negotiating agreements (MoUs) that promise respecting the security policies of others.
The MoUs typically mention only a small set of security classes, called \textit{transfer classes}, between which \textit{all} information exchange is managed. 
\textit{Modularity} and \textit{autonomy} are important: each organisation would wish to retain control over its own security policies and the ability to redefine them. 

We identified the elegant theory of Lagois Connections \cite{MELTON1994lagoisconnections} of order-preserving functions between security lattices as the appropriate framework, showing that they guarantee SIF, without the need for re-verifying the security of the application procedures in either of the domains, and confining the analysis to only the transfer classes involved in potential exchange of data.
This paper substantiates this by showing (i) how language-based techniques such as security type systems given by Volpano \textit{et al}. \cite{DBLP:journals/jcs/VolpanoIS96} can be adapted to a setting with different systems having distinct secure flow policies communicating data objects between each other; (ii) how results from Lagois theory provide a robust methodological framework for negotiating and maintaining secure agreements on information flow between autonomous organisations, even when either or both organisations change their security policies.
(ii) how decentralised secure flow systems such as that proposed by Myers \cite{myers-phd-tr-award} can be smoothly and conservatively extended to cross-organisational delegation and decentralisation.
Indeed, one way to view this work is that Lagois connections support a conservative extension of SIF analysis techniques developed on complete lattices to embrace structures involving lattices and morphisms between them. 

In \S\ref{sec:Bi-DirFlow}, we identify
intuitive requirements for secure bidirectional flow, present the definition of Lagois connections \cite{MELTON1994lagoisconnections}, and show that Lagois connections between the security lattices satisfy security and other requirements.
We include a brief account of how lattices and Lagois connections can be very succinctly represented to support efficient algorithms.

We present in \S\ref{sec:model} a minimal operational language consisting of a small set of \textit{atomic primitives} for effecting the transfer of data between domains.
The framework is simple and can be adapted for establishing secure connections between distributed systems at any level of abstraction (language, system, database, ...). 
We assume each domain uses \textit{atomic transactional operations} for object manipulation and intra-domain computation.
The primitives of our model include reliable and atomic communication between two systems, transferring object data in designated \textit{output} variables of one domain to designated \textit{input} variables of a specified security class in the other domain. 
To avoid interference between inter-domain communication and the computations within the domains, we assume that the sets of designated input and output variables are all mutually exclusive of one another, and also with the program/system variables used in the computations within each domain. 
Thus by design we avoid the usual suspects that cause interference and insecure transfer of data. 
The language should be seen as notation for execution sequences of atomic actions performed by concurrent  communicating systems.
So it  does not include conditional or iterative constructs, assuming these are absorbed within atomic intradomain transactions.
Thus, we do not have to concern ourselves with issues of implicit flows that arise due to branching structures (\textit{e.g.}, conditionals and  loops in programming language level security, pipeline mispredictions at the architectural level, etc.) 


The operational description of the language consists of the primitives
together with their execution rules (\S\ref{sec:operations}).
The correctness of our framework is demonstrated by expressing soundness (with respect to the operational semantics) of a type system (\S\ref{sec:typing}), stated in terms of the security lattices and their connecting functions.  
In particular, Theorem \ref{thm:soundness} shows the standard semantic property of \textit{non-interference} \cite{DBLP:conf/sp/GoguenM82a} in \textit{both domains} holds of all operational behaviours.
We adapt and \textit{extend} the approach taken by Volpano \textit{et al.} \cite{DBLP:journals/jcs/VolpanoIS96} to encompass systems coupled using the Lagois connection conditions, and (assuming atomicity of the data transfer operations) show that \textit{security is conserved}. 
Since our language is a minimal imperative model with atomic operations, and security types are exactly the security classes, our proof pares down the techniques of Volpano \textit{et al.} to the bare essentials. 

We revisit in \S\ref{sec:revisit-lagois} several results of  Melton \textit{et al.} \cite{MELTON1994lagoisconnections} that help us develop a methodical approach to finding and defining suitable MoUs for secure bidirectional flow, pictorially illustrating the systematic development on an example. 
In  \S\ref{sec:lagois-adjoint}, we use the fact that the morphisms of a Lagois connection uniquely determine each other to complete a MoU when given one side of a proposed mapping, by finding its \textit{Lagois adjoint}.
\S\ref{sec:lagois-abinitio} discusses a methodical approach to negotiating a viable secure MoU \textit{ab initio}. 
In \S\ref{sec:lagois-composition}, we use a compositionality result on Lagois connections to chain secure flows through a sequence of  organisations.
\S\ref{sec:maintaining-mou} tackles the issue of renegotiating and re-establishing a secure MoU when either party changes its security lattice. 
This is made possible by an decomposition result on Lagois connections discussed in \S\ref{sec:lagois-decomposition}.  
\S\ref{sec:changes} details the various techniques for rebuilding a secure MoU when changes are made to the security lattice, using appropriate results from 
\cite{MELTON1994lagoisconnections}

We then show that our Lagois connection framework readily accommodates decentralised flow control mechanisms within and across organisations, conservatively extending the model of Myers \cite{Myers1997-ss,myers-phd-tr-award}.
After a brief summary of the DLM framework in \S\ref{sec:DLM-summary}, we show in  Theorem \ref{Thm:LC-PH-LC-IFL} (\S\ref{sec:lagois-DLM-IFL}) that a Lagois connection between the principals hierarchies of two domains induces a Lagois connection between their corresponding lattices of labels.
A simple corollary, stated for static principals hierarchies that are connected by a Lagois connection, is that the declassification rule remains safe even with bidirectional information exchange between the domains.

In \S\ref{sec:related}, we briefly review some related work.
We conclude in \S\ref{sec:conclusion} with a discussion on our approach and directions for future work.

\textbf{Note}: This paper substantially expands on our earlier work \cite{BhardwajP2019}. 
Contents of that paper included here appear in the preliminaries in \S\ref{sec:Bi-DirFlow}, the technical results of \S\ref{sec:model} and the discussion on related work in \S\ref{sec:related}.

\section{Lagois Connections and All That}\label{sec:Bi-DirFlow}

\paragraph{Motivating Example} \ \ 
Consider a university $\rid{U}$ in which students study in semi-autonomously administered colleges (\textit{e.g.}, $\ld{C}$)   affiliated to the university.
A college has \textit{students}, \textit{faculty} members, \textit{deans}, all of whom work under a \textit{College Principal}.
The university has its own \textit{university professors}, a \textit{dean of colleges}, and a \textit{vice-chancellor} working under a \textit{chancellor}.
Students can take classes with both college faculty members and university professors.  
Assume that each institution has established its secure information flow mechanisms and policies, and information flows from $\ld{C}$ to $\rid{U}$, as shown by the blue arrows in Figure \ref{fig:WnD}.

\textit{Monotonicity} (or \textit{order-preservation}) of a function mapping security classes in $\ld{C}$ to $\rid{U}$ suffices for $\rid{U}$ to respect $\ld{C}$'s security policies.
However, as we showed earlier \cite{BhardwajP2019}, when flow is \textit{bidirectional}, composing order-preserving functions between $\ld{C}$ and $\rid{U}$ is insufficient.
Indeed, even a \textit{Galois connection} between the two domains does not ensure security (see Figure \ref{fig:Bi-WnD-yNotGC}).
While \textit{Galois insertions} ensure security, they require one of the functions to be \textit{surjective}, whereas in many situations, an organisation may not wish to expose its entire security class lattice.
Further, we do not wish information flowing between two domains to be reclassified in an overly restrictive manner that makes it inaccessible (we call these ``precision'' and ``convergence'' requirements).

\begin{figure}[t]
    \begin{minipage}[t]{.45\textwidth}
            \begin{tikzpicture}[framed,->,node distance=0.9cm,on grid]
    \title{W and D}
    \node(T2)   {\scriptsize$\top 2$};
    \node(T1)  [xshift=-3cm]  {\scriptsize$\top 1$};
    \node(Dir1) [below of = T1] {\scriptsize$CollegePrincipal$};
    \node(D1) [below left of = Dir1] {\scriptsize$Dean\ (F)$};
    \node(F1) [below of = D1] {\scriptsize$Faculty$};
    \node(DS1) [xshift=-2.3cm,yshift=-1.7cm] {\scriptsize$Dean\ (S)$};
    \node(S1) [below right of = F1] {\scriptsize$Student$};
    \node(B1)  [below of=S1]  {\scriptsize$\bot 1$};
    \node(Sec2)  [below of = T2] {\scriptsize$Chancellor$};
    \node(AS2)  [below of=Sec2] {\scriptsize$Vice\ Chancellor$};
    \node(Dir2)  [below of=AS2] {\scriptsize$Dean(Colleges)$};
    \node(E2)  [below of=Dir2] {\scriptsize$Univ.Fac.$}; 
    \node(B2)  [below of=E2]  {\scriptsize$\bot 2$};
    \draw [blue, thick] (S1) to (E2);
    \draw [OliveGreen, densely dashed, thick] (F1) to (E2);
    \draw [OliveGreen, densely dashed, thick] (B1) to (B2);
    \draw [OliveGreen, densely dashed, thick] (T1) to (T2);
    \draw [blue, thick] (Dir1) [bend left = 10] to (Dir2);
    \draw [OliveGreen, densely dashed, thick] (D1) [bend right = 15] to (Dir2);
    \draw [OliveGreen, densely dashed, thick] (DS1) to (Dir2);
    \draw [Fuchsia, ->] (Dir1) to (T1);
    \draw [Fuchsia, ->] (D1) to (Dir1);
    \draw [Fuchsia, ->] (F1) to (D1);
    \draw [Fuchsia, ->] (S1) to (F1);
    \draw [Fuchsia, ->] (S1) to (DS1);
    \draw [Fuchsia, ->] (DS1) to (Dir1);
    \draw [Fuchsia, ->] (B1) to (S1);
    \draw [Fuchsia, ->] (B2) to (E2);
    \draw [Fuchsia, ->] (E2) to (Dir2);
    \draw [Fuchsia, ->] (Dir2) to (AS2);
    \draw [Fuchsia, ->] (AS2) to (Sec2);
    \draw [Fuchsia, ->] (Sec2) to (T2);
    \draw (-3,-2.2)[blue] ellipse (1.4cm and 2.4cm);
    \draw (0,-2.2)[red] ellipse (1.2cm and 2.5cm);
    \end{tikzpicture}
    \caption{\small Unidirectional flow: If the solid blue arrows denote identified flows connecting important classes, then the dashed green arrows are constrained by monotonicity to lie between them. \label{fig:WnD}}
    \end{minipage}
    \quad \quad
    \begin{minipage}[t]{.45\textwidth}
   \begin{tikzpicture}[framed,->,node distance=1cm,on grid]
    \title{W and D}
    \node(T2)   {\scriptsize$\top 2$};
    \node(T1)  [xshift=-3cm]  {\scriptsize$\top 1$};
    \node(Dir1) [below of = T1] {\scriptsize$CollegePrincipal$};
    \node(D1) [below left of = Dir1] {\scriptsize$Dean\ (F)$};
    \node(F1) [below of = D1] {\scriptsize$Faculty$};
    \node(DS1) [xshift=-2.3cm,yshift=-1.7cm] {\scriptsize$Dean\ (S)$};
    \node(S1) [below right of = F1] {\scriptsize$Student$};
    \node(B1)  [below of=S1]  {\scriptsize$\bot 1$};
    \node(Sec2)  [below of = T2] {\scriptsize$Chancellor$};
    \node(AS2)  [below of=Sec2] {\scriptsize$Vice\ Chancellor$};
    \node(Dir2)  [below of=AS2] {\scriptsize$Dean(Colleges)$};
    \node(E2)  [below of=Dir2] {\scriptsize$Univ.Fac.$}; 
    \node(B2)  [below of=E2]  {\scriptsize$\bot 2$};
    \draw [OliveGreen,  thick] (S1) [bend left = 5] to (Dir2);
    \draw [red, densely dashdotted, thick] (Dir2) [bend left = 5] to (S1);
    \draw [red, densely dashdotted, thick] (F1) to (Dir2);
    \draw [OliveGreen,  thick] (B1) to (B2);
    \draw [OliveGreen,  thick, ->] (T1) [bend right=10] to (T2);
    \draw [brown,  thick, ->] (T2) [bend right=10] to (T1);
    \draw [OliveGreen,  thick] (Dir1) [bend left= 10] to (AS2);
    \draw [red, densely dashdotted, thick] (D1) to (Dir2);
    \draw [red, densely dashdotted, thick] (DS1) to (Dir2);
    \draw [Fuchsia, ->] (Dir1) to (T1);
    \draw [Fuchsia, ->] (D1) to (Dir1);
    \draw [Fuchsia, ->] (F1) to (D1);
    \draw [Fuchsia, ->] (S1) to (F1);
    \draw [Fuchsia, ->] (S1) to (DS1);
    \draw [Fuchsia, ->] (DS1) to (Dir1);
    \draw [Fuchsia, ->] (B1) to (S1);
    \draw [Fuchsia, ->] (B2) to (E2);
    \draw [Fuchsia, ->] (E2) to (Dir2);
    \draw [Fuchsia, ->] (Dir2) to (AS2);
    \draw [Fuchsia, ->] (AS2) to (Sec2);
    \draw [Fuchsia, ->] (Sec2) to (T2);
    \draw [brown,  thick] (B2) [bend left=20] to (B1);
    \draw [brown,  thick] (Sec2) to (T1);
    \draw [brown,  thick] (AS2) to (Dir1);
    \draw [brown,  thick] (E2) [bend left=10] to (S1);
    \draw (-3,-2.5)[blue] ellipse (1.4cm and 2.7cm);
    \draw (0,-2.5)[red] ellipse (1.2cm and 2.7cm);
    \end{tikzpicture}
    \caption{\small 
    The arrows between the domains define a Galois Connection.
    However, the red dash-dotted arrows highlight flow security violations when information can flow in both directions. 
    \label{fig:Bi-WnD-yNotGC}}
    \end{minipage}
   \label{fig:first-combined}
\end{figure}

\paragraph{Requirements} \ \ Accordingly,  we identified the following requirements \cite{BhardwajP2019} for viable secure information flow between two lattices $(\ld{L}, \ld{\sqsubseteq})$ and 
$(\rid{M}, \rd{\sqsubseteq})$ with order-preserving functions
$\rid{\alpha}: \ld{L} \rightarrow \rid{M}$ and 
$\ld{\gamma}: \rid{M} \rightarrow \ld{L}$.
\begin{itemize}
    \item \textit{Security:} 
    \[
\textbf{SC1}~~ \lambda \ld{l}.\ld{l} ~\ld{\sqsubseteq}~
\ld{\gamma} \circ \rid{\alpha}  
~~~~~~\hfill~~~~~~
\textbf{SC2} ~~  \lambda \rd{m}.\rd{m} ~ \rd{\sqsubseteq}~
\rid{\alpha} \circ \ld{\gamma} 
\]
   \item \textit{Precision:}
   Let $\rid{\alpha}[\ld{L}]$ and $\ld{\gamma}[\rid{M}]$ denote the images of
   set $\ld{L}$ under mapping $\rid{\alpha}$
   and set $\rid{M}$ under mapping $\ld{\gamma}$.
   \[
\begin{array}{c}
\textbf{PC1}~~\rid{\alpha}(\ld{l_1}) = \rid{\bigsqcup} ~ \{\rid{m_1} ~|~ \ld{\gamma}(\rid{m_1}) = \ld{l_1} \}, \; \; \forall \ld{l_1} \in \ld{\gamma}[\rid{M}]\\
\textbf{PC2} ~~\ld{\gamma}(\rid{m_1}) = \ld{\bigsqcup} ~ \{\ld{l_1} ~|~ \rid{\alpha}(\ld{l_1}) = \rid{m_1}\}, \; \; \forall \rid{m_1} \in \rid{\alpha}[\ld{L}]
\end{array}
\]
\item \textit{Convergence:} (\textbf{CC1} and \textbf{CC2})
 \textit{Fixed points} for the compositions $\ld{\gamma} \circ \rid{\alpha}$ and $\rid{\alpha} \circ \ld{\gamma}$
  are reached as low in the orderings $\ld{\sqsubseteq}$ and $\rd{\sqsubseteq}$ as possible.
\end{itemize}

\paragraph{Lagois Connections} \ \ 
We identified the elegant formulation of \textit{Lagois Connections} \cite{MELTON1994lagoisconnections} as an appropriate structure that satisfies these requirements:
\begin{definition}[Lagois Connection \cite{MELTON1994lagoisconnections}]
If $L = (\ld{L},\ld{\sqsubseteq})$ and $M = (\rid{M},\rd{\sqsubseteq})$ are two partially ordered sets, and $\rid{\alpha}: \ld{L} \rightarrow \rid{M}$ and $\ld{\gamma}: \rid{M} \rightarrow \ld{L}$ are order-preserving functions, then we call the quadruple $(\ld{L}, \rid{\alpha}, \ld{\gamma}, \rid{M})$ an {\em increasing} Lagois connection, if it satisfies the following properties:
\[
\begin{array}{llcll}
\textbf{LC1}~~ & \lambda \ld{l}.\ld{l} ~\ld{\sqsubseteq}~
\ld{\gamma} \circ \rid{\alpha}  
& ~~~~~~~~~~~ &
\textbf{LC2}~~ & \lambda \rd{m}.\rd{m} ~\rd{\sqsubseteq}~ 
\rid{\alpha} \circ \ld{\gamma} \\
\textbf{LC3}~~ &  \rid{\alpha} \circ  \ld{\gamma}  \circ \rid{\alpha} = \rid{\alpha}
& ~~~~~~~~~~~ &
\textbf{LC4}~~ & \ld{\gamma}  \circ \rid{\alpha}  \circ \ld{\gamma} = \ld{\gamma}
\end{array}
\]
\end{definition}

\textbf{LC3} ensures that $\ld{\gamma}(\rid{\alpha}(\ld{c_1}))$ is the least upper bound of all security classes in $\ld{C}$ that are mapped to the same security class, say $\rid{u_1} = \rid{\alpha}(\ld{c_1})$ in $\rid{U}$. 

Observe that Lagois connections are transposable: if
$(\ld{L}, \rid{\alpha}, \ld{\gamma}, \rid{M})$ is a Lagois connection, then so is 
$(\rid{M}, \ld{\gamma}, \rid{\alpha}, \ld{L})$.
Lagois connections are fundamentally different from Galois connections in that they relate two linked closure operators in two posets, as opposed to a linking a closure and an interior operator. 

We showed that Lagois connections satisfy the desired requirements: 
\begin{otheorem}[Theorem in \cite{BhardwajP2019}]\label{thm:secureconnection} 
Let $L = (\ld{L},\ld{\sqsubseteq}, \ld{\sqcup}, \ld{\sqcap})$ and $M = (\rid{M},\rd{\sqsubseteq}, \rd{\sqcup}, \rd{\sqcap})$ be two complete security class lattices, and let $\rid{\alpha}: \ld{L} \rightarrow \rid{M}$ and $\ld{\gamma}: \rid{M} \rightarrow \ld{L}$ be  order-preserving functions. 
Then the flow of information permitted by $\rid{\alpha}$, $\ld{\gamma}$ satisfies conditions \textbf{SC1}, \textbf{SC2}, \textbf{PC1}, \textbf{PC2}, \textbf{CC1} and \textbf{CC2}
if $(\ld{L}, \rid{\alpha}, \ld{\gamma}, \rid{M})$ is an increasing Lagois connection. 
\end{otheorem}

In the following discussion, let $(\ld{L}, \rid{\alpha}, \ld{\gamma}, \rid{M})$ be a Lagois connection.
Let $\rid{\alpha}[\ld{L}]$ and $\ld{\gamma}[\rid{M}]$ refer to the images of the order-preserving functions $\rid{\alpha}$ and $\ld{\gamma}$, respectively.  
The images $\ld{\gamma}[\rid{M}]$ and $\rid{\alpha}[\ld{L}]$ are in fact isomorphic lattices.
For all $\rid{m} \in \rid{\alpha}[\ld{L}]$ and $\ld{l} \in \ld{\gamma}[\rid{M}]$, $\ld{\gamma}(\rid{m})$ and $\rid{\alpha}(\ld{l})$ exist.
\begin{oproposition}[Proposition 3.7 in \cite{MELTON1994lagoisconnections}]\label{prop:largest-pre}
Let $\rid{m} \in \rid{\alpha}[\ld{L}]$ and $\ld{l} \in \ld{\gamma}[\rid{M}]$. 
Then $\rid{\alpha}^{-1}(\rid{m})$ has a largest member, which is $\ld{\gamma}(\rid{m})$, and $\ld{\gamma}^{-1}(\ld{l})$ has a largest member, which is $\rid{\alpha}(\ld{l})$. 
\end{oproposition}
We call these dominating members of the pre-images of $\rid{\alpha}$ and $\ld{\gamma}$ \textit{budpoints}.
Indeed:
\begin{align}
    \ld{\gamma}(\rid{\alpha}(\ld{l})) ~=~ 
    \ld{\sqcap} \{ \ld{l^*} \in \ld{\gamma}[\rid{M}] ~|~  
    \ld{l} ~\ld{\sqsubseteq}~ \ld{l^*} \}, \label{TIGHT1} \\
     \rid{\alpha}( \ld{\gamma}(\rid{m})) ~=~
     \rid{\sqcap} \{ \rid{m^*} \in \rid{\alpha}[\ld{L}] ~|~
        \rid{m} ~\rd{\sqsubseteq}~ \rid{m^*} \}.
        \label{TIGHT2}
\end{align}

Let  $\rid{\thicksim_M}$ and $\ld{\thicksim_L}$ be the equivalence relations induced by the functions $\ld{\gamma}$ and $\rid{\alpha}$.
$\ld{L^*} = \ld{\gamma}[\rid{\alpha}[\ld{L}]] = \ld{\gamma}[\rid{M}]$  and
$\rid{M^*} = \rid{\alpha}[\ld{\gamma}[\rid{M}]] = \rid{\alpha}[\ld{L}]$ define a system of representatives for $\ld{\thicksim_L}$ and $\rid{\thicksim_M}$. 
Element $\rid{m^*} = \rid{\alpha}(\ld{\gamma}(\rd{m}))$ in $\rid{M^*}$, which is a \textit{budpoint}, acts as the representative of the equivalence class $[\rd{m}]$ in the following sense:
\begin{align}
    \textit{if}~ \rid{m}\in \rid{M} ~\textit{and}~ \rid{m^*} \in \rid{M^*} ~\textit{with}~ \rid{m} ~\rid{\thicksim_M}~ \rid{m^*}~\textit{then}~ \rid{m} ~\rd{\sqsubseteq}~ \rid{m^*}
\end{align}
Symmetrically, $\ld{L^*} = \ld{\gamma}[\rid{\alpha}[\ld{L}]] = \ld{\gamma}[\rid{M}]$ defines a system of representatives for $\ld{\thicksim_L}$.
These budpoints play a significant role in delineating the connection between the transfer classes in the two lattices.
Further, Proposition \ref{prop:meets} shows that these budpoints are closed under meets. 
This property enables us to confine our analysis to just these security classes when reasoning about bidirectional flows.
\begin{proposition}[Proposition 3.11 in \cite{MELTON1994lagoisconnections}]\label{prop:meets}
\label{prop:meet-existence}
If $\ld{A} \subseteq \ld{\gamma}[\rid{M}]$, then 
\begin{enumerate}
    \item the meet of $\ld{A}$ in $\ld{\gamma}[\rid{M}]$ exists if and only if the meet of $\ld{A}$ in $\ld{L}$ exists, and whenever either exists, they are equal.
    \item the join $\ld{\hat{a}}$ of $\ld{A}$ in $\ld{\gamma}[\rid{M}]$ exists if the join $\ld{\check{a}}$ of $\ld{A}$ in $\ld{L}$ exists, and in this case $\ld{\hat{a}} = \ld{\gamma}(\rid{\alpha}(\ld{\check{a}}))$.
\end{enumerate}
\end{proposition}


\subsection{Algorithmic Issues} \label{sec:complxley}

We propose using a recently proposed succinct representation for lattices \cite{munro2020space}, in which order-testing comparisons can be answered in $\mathcal{O}(1)$ time, and the meet or join of two elements in $\mathcal{O}(n^{3/4})$, where $n$ is the number of elements in the lattice.
The data structure occupies $\mathcal{O}(n^{3/2}\ log{}\ n)$
bits of space, with pre-processing time  $\mathcal{O}(n^{2})$.
The functions $\rid{\alpha}, \ld{\gamma}$ and their inverses are represented using hash-based dictionaries, and equivalence classes induced by $\rid{\alpha}, \ld{\gamma}$ are represented using the union-find data structure. 
Thus, the total space for a Lagois connection representation is $\mathcal{O}(n^{3/2})$, where $n$ is $max(n_1,n_2)$, with $n_1$ and $n_2$ being the number of elements in $\ld{L}$ and $\rid{M}$ respectively.

Order-comparisons within a domain are checked using the succinct data structure. 
Cross-domain flows, e.g., between $\ld{x} \in \ld{L}$ and $\rid{y} \in \rid{M}$, are permitted if and only if $(\exists \rid{z}\in \rid{M}. \rid{\alpha}(\ld{x}) = \rid{z}\ and\ \rid{z} \rd{\sqsubseteq} \rid{y})$.
Since looking up $\rid{\alpha}(\ld{x})$ or $\ld{\gamma}(\rid{y})$ in a hashing-based dictionary can be done in $\mathcal{O}(1)$ time, all
order-comparisons can be performed in $\mathcal{O}(1)$.

Testing \textbf{LC1} and \textbf{LC3} takes $\mathcal{O}(n_1)$
time, and \textbf{LC2} and \textbf{LC3}  $\mathcal{O}(n_2)$
time.
Thus checking whether $(\ld{L},\rid{\alpha},\ld{\gamma}, \rid{M})$ is a Lagois connection takes $\mathcal{O}(n)$ time, where $n = max(n_1,n_2)$. 

If the transfer classes have been identified and these are far fewer than the number of lattice points, we can further optimise the pre-processing and the data structure (since the two sets of transfer classes are order-isomorphic to each other as we shall in Theorem \ref{lemma:decompose} of \S\ref{sec:lagois-decomposition}).

\section{An Operational Model}\label{sec:model}

We present an operational model consisting of a \textit{language} and its operational semantics, and then a \textit{security type system}, which we show to be sound in that well-typed programs exhibit non-interference \cite{DBLP:conf/sp/GoguenM82a}.
The objective of this section is to show that under reasonable assumptions of atomicity and isolation, the Lagois connection framework allows SIF analyses performed within systems to be lifted systematically to concurrent systems exchanging information. 

\subsection{Computational Model.}\label{sec:operations}
Assume two organisations $\ld{L}$ and $\rid{M}$ which have their own IFMs want to share data with each other.
The two domains comprise storage objects $\ld{Z}$ and $\rd{Z}$ respectively, which are manipulated using their own sets of \textit{atomic} transactional operations, ranged over by $\ld{t}$ and $\rd{t}$ respectively. 
We further assume that these transactions within each domain are internally secure with respect to their flow models, and have no insecure or interfering interactions with the environment. 
Thus, we are agnostic to the level of abstraction of the systems we aim to connect securely, and since our approach treats the application domains as ``black boxes'', it is readily adaptable to any level of discourse (language, system, OS, database) found in the security literature.

We extend these operations with a minimal set of operations to transfer data between the two domains.
To avoid any concurrency effects, interference or race conditions arising from inter-domain transfer, we augment the storage objects of both domains with a fresh set of \textit{export} and \textit{import} variables into/from which the data of the domain objects can be copied \textit{atomically}.  
We designate these sets $\ld{X}, \rd{X}$ as the respective \textit{export} variables, and $\ld{Y}, \rd{Y}$ as the respective \textit{import} variables, with the corresponding variable instances written as $\ld{x_i}$, $\rd{x_i}$ and $\ld{y_i}$, $\rd{y_i}$.
These export and import variables form mutually disjoint sets, and are distinct from any extant domain objects manipulated by the applications within a domain. 
These variables are used exclusively for transfer, and are manipulated atomically.  
We let $\ld{w_i}$ range over all variables in $\ld{N} ~=~\ld{Z} \cup \ld{X} \cup \ld{Y}$ (respectively $\rd{w_i}$ over $\rd{N} ~=~ \rd{Z} \cup \rd{X} \cup \rd{Y}$).
Domain objects are copied \textit{to export} variables and \textit{from import} variables by special operations $\ld{rd}(\ld{z}, \ld{y})$ and $\ld{wr}(\ld{x}, \ld{z})$ (and $\rd{rd}(\rd{z}, \rd{y})$ and  $\rd{wr}(\rd{x}, \rd{z})$ in the other domain).
We assume \textit{atomic transfer} operations (\textit{trusted by both domains}) $T_{RL}, T_{LR}$ that copy data from the export variables of one domain to the import variables of the other domain as the only mechanism for inter-domain flow of data.
Let  ``phrase'' $p$ denote a command in either domain or a transfer operation, and let $s$ be any (empty or non-empty) sequence of phrases.
\textbf{Note}: This language should be understood as a notation for describing a distributed system's execution sequences involving computation, communication, input and output, rather than as a programming language.  
Thus, we only need to consider sequences of atomic actions.
Hence the absence of constructs such as conditionals, iteration or repetition.
Further, the importance of atomicity of the computational steps and communication-related operations should be evident.
We will later see that the ``types'' are exactly the security classes of the IFMs, and so there are no constructions such as cartesian product, records and function types.
\[
\begin{array}{c}
\text{(command)} ~~~
    \ld{c} ::=  \ld{t} ~|~ \ld{rd}(\ld{z}, \ld{y}) ~|~  \ld{wr}(\ld{x}, \ld{z}) ~~~~~~ 
    \rd{c} ::=  
    \rd{t} ~|~ \rd{rd}(\rd{z}, \rd{y}) ~|~  \rd{wr}(\rd{x}, \rd{z}) \\
    \text{(phrase)} ~~~ p ::= T_{RL}(\rd{x},\ld{y}) ~|~ T_{LR}(\ld{x},\rd{y}) ~|~ 
    \ld{c} ~|~ \rd{c}  ~~~\hfill~~~  \text{(seq)}~~~ s ::= \epsilon ~|~ s_1 ; p \\ 
\end{array}
\]

\begin{figure}
\[
\begin{array}{c}
\inferrule* [Left = \ld{T}]{\ld{\mu} \vdash \ld{t} \Rightarrow \ld{\nu}
}{\langle \ld{\mu},\rd{\mu} \rangle \vdash \ld{t} \Rightarrow \langle \ld{\nu}, \rd{\mu} \rangle}
~~~~\hfill~~~~
\inferrule* [Left = \rd{T}]{\rd{\mu} \vdash \rd{t} \Rightarrow \rd{\nu}
}{\langle \ld{\mu},\rd{\mu} \rangle \vdash \rd{t} \Rightarrow \langle \ld{\mu}, \rd{\nu} \rangle} \\[1ex]
\inferrule*[Left = \ld{Wr}]{ 
}{\langle \ld{\mu},\rd{\mu} \rangle \vdash \ld{wr}(\ld{x},\ld{z}) \Rightarrow \langle \ld{\mu}[\ld{x} := \ld{\mu}(\ld{z})], \rd{\mu} \rangle}
\\[1ex]
\inferrule*[Left = \rd{Wr}]{ 
}{\langle \ld{\mu},\rd{\mu} \rangle \vdash \rd{wr}(\rd{x},\rd{z}) \Rightarrow \langle \ld{\mu}, \rd{\mu}[\rd{x} := \rd{\mu}(\rd{z})] \rangle}
\\[1ex]
\inferrule*[Left = \ld{Rd}]{ 
}{\langle \ld{\mu},\rd{\mu} \rangle \vdash \ld{rd}(\ld{z},\ld{y}) \Rightarrow \langle \ld{\mu}[\ld{z} := \ld{\mu}(\ld{y})], \rd{\mu} \rangle} \\[1ex]
\inferrule*[Left = \rd{Rd}]{ 
}{\langle \ld{\mu},\rd{\mu} \rangle \vdash \rd{rd}(\rd{z},\rd{y}) \Rightarrow \langle \ld{\mu}, \rd{\mu}[\rd{z} := \rd{\mu}(\rd{y})] \rangle} \\[1ex]
\inferrule*[Left = Trl]{ 
}{\langle \ld{\mu},\rd{\mu} \rangle \vdash T_{RL}(\ld{y},\rd{x}) \Rightarrow 
\langle \ld{\mu}[\ld{y} := \rd{\mu}(\rd{x})],\rd{\mu} \rangle} \\[1ex]
\inferrule*[Left = Tlr]{ 
}{\langle \ld{\mu},\rd{\mu} \rangle \vdash T_{LR}(\rd{y},\ld{x}) \Rightarrow 
\langle \ld{\mu},\rd{\mu}[\rd{y} := \ld{\mu}(\ld{x})] \rangle} \\[1ex]
\inferrule* [Left = Seq0]{ }{
\langle \ld{\mu},\rd{\mu} \rangle \vdash \epsilon \Rightarrow^* \langle \ld{\mu},\rd{\mu} \rangle} \\[1ex]
\inferrule* [Left = SeqS]{\langle \ld{\mu},\rd{\mu} \rangle \vdash s_1 \Rightarrow^* \langle \ld{\mu_1},\rd{\mu_1} \rangle, ~~~
\langle \ld{\mu_1},\rd{\mu_1} \rangle \vdash p \Rightarrow  \langle \ld{\mu_2},\rd{\mu_2} \rangle}{\langle \ld{\mu},\rd{\mu} \rangle \vdash s_1 ; p \Rightarrow^* \langle \ld{\mu_2},\rd{\mu_2} \rangle}
\end{array}
\]
\caption{Execution Rules}
    \label{fig:evalrules11}
\end{figure}

A \textit{store} (typically $\ld{\mu}, \ld{\nu}, \rd{\mu}, \rd{\nu}$) is a finite-domain function from variables to a set of values (not further specified). 
We write, \textit{e.g.}, $\ld{\mu}(\ld{w})$ for the contents of the store $\ld{\mu}$ at variable $\ld{w}$, and $\ld{\mu}[\ld{w} := \rd{\mu}(\rd{w})]$ for the store 
that is the same as $\ld{\mu}$ everywhere except at variable $\ld{w}$,
where it now takes value $\rd{\mu}(\rd{w})$.

The rules specifying execution of commands are given in Fig. \ref{fig:evalrules11}.
Assuming the specification of intradomain transactions (\textit{\ld{t}, \rd{t}}) of the form
$\ld{\mu} \vdash \ld{t} \implies \ld{\nu}$ and $\rd{\mu} \vdash \rd{t} \implies \rd{\nu}$, our rules allow us to specify judgments of the form
$\langle \ld{\mu}, \rd{\mu} \rangle \vdash p \implies \langle \ld{\nu}, \rd{\nu} \rangle$ for phrases, and the reflexive-transitive closure for sequences of phrases.
Note that phrase execution occurs \textit{atomically}, and the intra-domain transactions, as well as copying to and from the export/import variables affect the store in only one domain, whereas the \textit{atomic transfer} is only between export variables of one domain and the import variables of the other.

\subsection{Typing Rules}\label{sec:typing}
Let the two domains have the respective different IFMs:
\[ FM_L = \langle \ld{N}, \ld{P}, \ld{SC}, \ld{\sqcup}, \ld{\sqsubseteq} \rangle  ~~~\hfill~~~ FM_M = \langle \rd{N}, \rd{P}, \rd{SC}, \rid{\sqcup}, \rd{\sqsubseteq} \rangle, \]
such that the flow policies in both are defined over different sets of security classes $\ld{SC}$ and $\rd{SC}$.\footnote{Without loss of generality, we assume that $\ld{SC} \cap \rd{SC} = \emptyset$, since we can suitably rename security classes.} 
 
The (security) types of the core language are as follows.

Metavariables $\ld{l}$ and $\rd{m}$ range over the sets of security classes, $\ld{SC}$ and $\rd{SC}$ respectively, which are partially ordered by $\ld{\sqsubseteq}$ and $\rd{\sqsubseteq}$. 
Note that,  with the language being minimal, there are no other type constructions. 
A type assignment $\ld{\lambda}$ is a finite-domain function from variables $\ld{N}$ to $\ld{SC}$ (respectively, $\rd{\lambda}$ from $\rd{N}$ to $\rd{SC}$).
The important restriction we place on $\ld{\lambda}$
and $\rd{\lambda}$ is that they map export and import variables $\ld{X}, \ld{Y}, \rd{X}, \rd{Y}$ only to points in the security lattices $\ld{SC}$ and $\rd{SC}$ respectively which are in the domains of $\ld{\gamma}$ and $\rid{\alpha}$, \textit{i.e.}, these points participate in the Lagois connection.  
Intuitively, a variable $w$ mapped to security class $\ld{l}$ can store information of security class $\ld{l}$ or lower.
The type system works with respect to a given type assignment.  Given a security level, \textit{e.g.}, $\ld{l}$, the typing rules track for each command \textit{within that domain} whether all written-to variables in that domain are of security classes ``above'' $\ld{l}$, and additionally for transactions within a domain, they ensure  ``simple security'', \textit{i.e.}, that  all variables which may have been read belong to security classes ``below'' $\ld{l}$.
We assume for the transactions within a domain, \textit{e.g.}, $\ld{L}$, we already have a security type system that will give us judgments of the form $\ld{\lambda} \vdash \ld{t}: \ld{l}$.
The novelty of our approach is to extend this framework to work over two connected domains, \textit{i.e.}, given implicit security levels of the contexts in the respective domains.
Cross-domain transfers will require pairing such judgments, and thus our type system will have judgments of the form 
\[ \langle \ld{\lambda}, \rd{\lambda} \rangle \vdash p: \langle 
\ld{l}, \rd{m} \rangle \]

We introduce a set of typing rules for the core language, given in Fig. \ref{fig:syntax-type-rules11}.  
In many of the rules, the type for one of the domains is not constrained by the rule, and so any suitable type may be chosen as determined by the context,  \textit{e.g.}, $\rd{m}$ in the rules \ld{\sc Tt},
\ld{\sc Trd}, \ld{\sc Twr} and $\ld{TT_{RL}}$,
and both $\ld{l}$ and $\rd{m}$ in {\sc Com0}.

\begin{figure}[!ht]
\centering
\[
\begin{array}{c}

\inferrule*[Left = \ld{Tt}]{~}{
\langle \ld{\lambda}, \rd{\lambda} \rangle \vdash \ld{t}: \langle \ld{l}, \rd{m} \rangle} ~\text{if for all $\ld{z}$ assigned in $\ld{t}$, }
\ld{l} ~\ld{\sqsubseteq}~ \ld{\lambda}(\ld{z}) \\
~~~~~~~~~~~~~~~~~~~~\hfill
\text{ \& for all  $\ld{z_1}$  read in $\ld{t}$, }
\ld{\lambda}(\ld{z_1}) ~\ld{\sqsubseteq}~ \ld{l}
\\
\inferrule*[Left = \rd{Tt}]{~}{
\langle \ld{\lambda}, \rd{\lambda} \rangle \vdash \rd{t}: \langle \ld{l}, \rd{m} \rangle} ~\text{if for all $\rd{z}$ assigned in $\rd{t}$, } 
\rd{m} ~\rd{\sqsubseteq}~ \rd{\lambda}(\rd{z}) \\
~~~~~~~~~~~~~~~~~~~~\hfill
\text{ \& for all } \rd{z_1} \text{ read in $\rd{t}$, }
\rd{\lambda}(\rd{z_1}) ~\rd{\sqsubseteq}~ \rd{m}
\\[1ex]
\inferrule*[Left = \ld{Trd}]{\ld{\lambda}(\ld{y}) ~\ld{\sqsubseteq}~ \ld{\lambda}(\ld{z})}{
\langle \ld{\lambda}, \rd{\lambda} \rangle \vdash \ld{rd}(\ld{z},\ld{y}): 
\langle \ld{\lambda}(\ld{z}), \rd{m} \rangle
} 
\\
\inferrule*[Left = \rd{Trd}]{\rd{\lambda}(\rd{y}) ~\rd{\sqsubseteq}~ \rd{\lambda}(\rd{z})}{
\langle \ld{\lambda}, \rd{\lambda} \rangle \vdash \rd{rd}(\rd{z},\rd{y}): 
\langle \ld{l}, \rd{\lambda}(\rd{z}) \rangle
} 
\\[1ex]
\inferrule*[Left = \ld{Twr}]{\ld{\lambda}(\ld{z}) ~\ld{\sqsubseteq}~ \ld{\lambda}(\ld{x})}{
\langle \ld{\lambda}, \rd{\lambda} \rangle \vdash \ld{wr}(\ld{x},\ld{z}): 
\langle \ld{\lambda}(\ld{x}), \rd{m} \rangle} 
\\
\inferrule*[Left = \rd{Twr}]{\rd{\lambda}(\rd{z}) ~\rd{\sqsubseteq}~ \rd{\lambda}(\rd{x})}{
\langle \ld{\lambda}, \rd{\lambda} \rangle \vdash \rd{wr}(\rd{x},\rd{z}): 
\langle \ld{l}, \rd{\lambda}(\rd{x}) \rangle} \\[1ex]
\inferrule*[Left = \ld{$TT_{RL}$}]{ 
\ld{\gamma}(\rd{\lambda}(\rd{x})) ~\ld{\sqsubseteq}~ \ld{\lambda}(\ld{y})
}{
\langle \ld{\lambda}, \rd{\lambda} \rangle \vdash T_{RL}(\ld{y},\rd{x}): 
\langle \ld{\lambda}(\ld{y}), \rd{\lambda}(\rd{x}) \rangle} 
\\
\inferrule*[Left = \rd{$TT_{LR}$}]{ 
\rid{\alpha}(\ld{\lambda}(\ld{x})) ~\rd{\sqsubseteq}~ \rd{\lambda}(\rd{y})
}{
\langle \ld{\lambda}, \rd{\lambda} \rangle \vdash T_{LR}(\rd{y},\ld{x}): 
\langle \ld{\lambda}(\ld{x}), \rd{\lambda}(\rd{y}) \rangle }
\\[1ex]
\inferrule*[Left = Com0]{ }{
\langle \ld{\lambda}, \rd{\lambda} \rangle \vdash \epsilon: \langle \ld{l}, \rd{m} \rangle} 
\\
\inferrule*[Left = ComS]{
\langle \ld{\lambda}, \rd{\lambda} \rangle \vdash p: \langle \ld{l_1}, \rd{m_1} \rangle
~~~~~\langle \ld{\lambda}, \rd{\lambda} \rangle \vdash s: \langle \ld{l}, \rd{m} \rangle
}{\langle \ld{\lambda}, \rd{\lambda} \rangle \vdash s; p: 
\langle \ld{l_1} \ld{\sqcap} \ld{l} , \rd{m_1} \rid{\sqcap} \rd{m} \rangle}
\end{array} 
    \]
    \caption{Typing Rules}
    \label{fig:syntax-type-rules11}
\end{figure}

For transactions \textit{e.g.}, $\ld{t}$ entirely within domain $\ld{L}$, the typing rule \ld{\sc Tt} constrains the type in the left domain to be at a level $\ld{l}$ that dominates all variables read in
$\ld{t}$, and which is dominated by all variables written to in $\ld{t}$, but places no constraints on the type $\rd{m}$ in the other domain $\rid{M}$.
In the rule \ld{\sc Trd}, since a value in import variable $\ld{y}$ is copied to the variable $\ld{z}$, we have $\ld{\lambda}(\ld{y}) ~\ld{\sqsubseteq}~ \ld{\lambda}(\ld{z})$, and the type in the domain $\ld{L}$ is $\ld{\lambda}(\ld{z})$ with no constraint on the type $\rd{m}$ in the other domain.
Conversely, in the rule \ld{\sc Twr}, since a value in variable $\ld{z}$ is copied to the export variable $\ld{x}$, we have $\ld{\lambda}(\ld{z}) ~\ld{\sqsubseteq}~ \ld{\lambda}(\ld{x})$, and the type in the domain $\ld{L}$ is $\ld{\lambda}(\ld{x})$ with no constraint on the type $\rd{m}$ in the other domain.
In the rule $\ld{TT_{RL}}$, since the contents of
a variable $\rd{x}$ in domain $\rid{M}$ are copied into a variable $\ld{y}$ in domain $\ld{L}$, we require $\ld{\gamma}(\rd{\lambda}(\rd{x})) ~\ld{\sqsubseteq}~ \ld{\lambda}(\ld{y})$, and constrain the type in domain $\ld{L}$ to $\ld{\lambda}(\ld{y})$.  
The constraint in the other domain is unimportant (but for the sake of convenience, we peg it at $\rd{\lambda}(\rd{x})$).
Finally, for the types of sequences of phrases, we take the meets of the collected types in each domain respectively, so that we can guarantee that no variable of type lower than these meets has been written into during the sequence.
Note that Proposition \ref{prop:meet-existence} ensures that these types have the desired properties for participating in the Lagois connection. 

\subsection{Soundness} \label{sec:soundness}
We now establish soundness of our scheme by showing a non-interference theorem with respect to the operational semantics and the type system built on the security lattices.
This theorem may be viewed as a conservative adaptation (to a minimal secure data transfer framework in a Lagois-connected pair of domains) of the main result of Volpano \textit{et al.} \cite{DBLP:journals/jcs/VolpanoIS96}.

We assume that underlying base transactional languages in each of the domains have the following simple property (stated for $\ld{L}$, but an analogous property is assumed for $\rid{M}$).
Within each transaction $\ld{t}$, for each assignment of an expression $\ld{e}$ to any variable $\ld{z}$,  the following holds:
If $\ld{\mu}$, $\ld{\nu}$ are two stores such that
for all $\ld{w} \in vars(\ld{e})$, we have $\ld{\mu}(\ld{w}) = \ld{\nu}(\ld{w})$, then after executing the assignment, we will get $\ld{\mu}(\ld{z}) = \ld{\nu}(\ld{z})$.  
That is, if two stores are equal for all variables appearing in the expression $\ld{e}$, then the value assigned to the variable $\ld{z}$ will be the same.
This assumption plays the r\^{o}le of ``Simple Security'' of expressions in \cite{DBLP:journals/jcs/VolpanoIS96} in the proof of the main theorem.
The type system plays the r\^{o}le of ``Confinement''.
We start with two obvious lemmas about the operational semantics, namely preservation of domains, and a ``frame'' lemma:
\begin{lemma}[Domain preservation]\label{lemma:equaldomainoneval}
If $\langle \ld{\mu}, \rd{\mu} \rangle \vdash s \Rightarrow^* \langle \ld{\mu_1}, \rd{\mu_1} \rangle$, then $dom(\ld{\mu}) = dom(\ld{\mu_1})$, and $dom(\rd{\mu}) = dom(\rd{\mu_1})$.
\end{lemma}
\begin{proof}
By induction on the length of the derivation of $\langle \mu, \rd{\mu} \rangle \vdash s \Rightarrow^* \langle \mu_1, \rd{\mu_1} \rangle$.
\end{proof}
\begin{lemma}[Frame]\label{lemma:notassigned}
If $\langle \mu, \rd{\mu} \rangle \vdash s \Rightarrow^* \langle \mu_1, \rd{\mu_1} \rangle, w \in dom(\mu) \cup dom(\rd{\mu})$, and $w$ is \textit{not} assigned to in $s$, then $\ld{\mu}(w) = \ld{\mu_1}(w)$ and
$\rd{\mu}(w) = \rd{\mu_1}(w)$.
\end{lemma}
\begin{proof}
By induction on the length of the derivation of $\langle \ld{\mu}, \rd{\mu} \rangle \vdash s \Rightarrow^* \langle \ld{\mu_1}, \rd{\mu_1} \rangle$.
\end{proof}

The main result assumes an ``adversary'' that operates at a security level $\ld{l}$ in domain $\ld{L}$ and at security level $\rd{m}$ in domain $\rid{M}$.  
Note however, that these two levels are interconnected by the monotone functions $\rid{\alpha}: \ld{L} \rightarrow \rid{M}$ and
$\ld{\gamma}: \rid{M} \rightarrow \ld{L}$, since these levels are connected by the ability of information at one level in one domain to flow to the other level in the other domain.
The following theorem says that if (a) a sequence of phrases is well-typed, and (b,c) we start its execution in two store configurations that are (e) indistinguishable with respect to all objects having security class below $\ld{l}$ and $\rd{m}$ in the respective domains, then the corresponding resulting stores after execution continue to remain indistinguishable on all variables with security classes below these adversarial levels.  
\begin{theorem}[Type Soundness]\label{thm:soundness}
Suppose $\ld{l}, \rd{m}$ are the ``adversarial'' type levels in the respective domains, which satisfy the condition $\ld{l} = \ld{\gamma}(\rd{m})$ and
$\rd{m} = \rid{\alpha}(\ld{l})$.  
Let 
\begin{enumerate}[label=(\alph*)]
    \item \label{assume:1}$\langle \ld{\lambda}, \rd{\lambda} \rangle \vdash s: 
    \langle \ld{l_0}, \rd{m_0}\rangle$; ~~~
    ($s$ has security type $\langle \ld{l_0}, \rd{m_0}\rangle$)
    \item \label{assume:2}$\langle \ld{\mu}, \rd{\mu} \rangle \vdash s \Rightarrow^* \langle \ld{\mu_f}, \rd{\mu_f} \rangle$;
    ~~(execution of $s$ starting from $\langle \ld{\mu}, \rd{\mu} \rangle $)
    \item \label{assume:3}$\langle \ld{\nu}, \rd{\nu} \rangle \vdash s \Rightarrow^* \langle \ld{\nu_f}, \rd{\nu_f} \rangle$;
        ~~(execution of $s$ starting from $\langle \ld{\nu}, \rd{\nu} \rangle $)
    \item \label{assume:4}$dom(\ld{\mu}) = dom(\ld{\nu}) = dom(\ld{\lambda})$ and
    $dom(\rd{\mu}) = dom(\rd{\nu}) = dom(\rd{\lambda})$;
    \item \label{assume:5} $\ld{\mu}(\ld{w}) = \ld{\nu}(\ld{w})$ for all
    $\ld{w}$ such that $\ld{\lambda}(\ld{w}) ~\ld{\sqsubseteq}~ \ld{l}$, and
    $\rd{\mu}(\rd{w}) = \rd{\nu}(\rd{w})$ for all
    $\rd{w}$ such that $\rd{\lambda}(\rd{w}) ~\rd{\sqsubseteq}~ \rd{m}$.

\end{enumerate}
Then 
$\ld{\mu_f}(\ld{w}) = \ld{\nu_f}(\ld{w})$ for all
    $\ld{w}$ such that $\ld{\lambda}(\ld{w}) ~\ld{\sqsubseteq}~ \ld{l}$, and
    $\rd{\mu_f}(\rd{w}) = \rd{\nu_f}(\rd{w})$ for all
    $\rd{w}$ such that $\rd{\lambda}(\rd{w}) ~\rd{\sqsubseteq}~ \rd{m}$.

\end{theorem}
\begin{proof}
By induction on the length of sequence $s$.  
The base case is vacuously true. 
We now consider a sequence $s_1; p$. 
$\langle \ld{\mu},\rd{\mu} \rangle \vdash s_1 \Rightarrow^* \langle \ld{\mu_1},\rd{\mu_1} \rangle$ and 
$\langle \ld{\mu_1},\rd{\mu_1} \rangle \vdash p \Rightarrow  \langle \ld{\mu_f},\rd{\mu_f} \rangle$
and
$\langle \ld{\nu},\rd{\nu} \rangle \vdash s_1 \Rightarrow^* \langle \ld{\nu_1},\rd{\nu_1} \rangle$ and 
$\langle \ld{\nu_1},\rd{\nu_1} \rangle \vdash p \Rightarrow  \langle \ld{\nu_f},\rd{\nu_f} \rangle$
By induction hypothesis applied to $s_1$, we have 
$\ld{\mu_1}(\ld{w}) = \ld{\nu_1}(\ld{w})$ for all
$\ld{w}$ such that $\ld{\lambda(\ld{w})} ~\ld{\sqsubseteq}~\ld{l}$,
and 
$\rd{\mu_1}(\rd{w}) = \rd{\nu_1}(\rd{w})$ for all
$\rd{w}$ such that $\rd{\lambda}(\rd{w}) ~\rd{\sqsubseteq}~ \rd{m}$.
    
Let  $\langle \ld{\lambda}, \rd{\lambda} \rangle \vdash s_1: \langle \ld{l_s}, \rd{m_s} \rangle$, and $\langle \ld{\lambda}, \rd{\lambda} \rangle \vdash p: \langle \ld{l_p}, \rd{m_p} \rangle$.
We examine four cases for $p$ (the remaining cases are symmetrical). \\
\textbf{Case} $p$ is $\ld{t}$: \ \ \ 
Consider any $\ld{w}$ such that 
$\ld{\lambda}(\ld{w}) ~\ld{\sqsubseteq}~\ld{l}$.
If $\ld{w} \in \ld{X} \cup \ld{Y}$ (\textit{i.e.}, it doesn't appear in $\ld{t}$), or if $\ld{w} \in \ld{Z}$ but is not assigned to in $\ld{t}$, then 
by Lemma \ref{lemma:notassigned} and the induction hypothesis, $\ld{\mu_f}(\ld{w}) =  
\ld{\mu_1}(\ld{w}) = \ld{\nu_1}(\ld{w}) = 
\ld{\nu_f}(\ld{w})$. \\
Now suppose $\ld{z}$ is assigned to in $\ld{t}$.
From the condition $\langle \ld{\lambda}, \rd{\lambda} \rangle \vdash p: \langle \ld{l_p}, \rd{m_p} \rangle$, we know that
for all $\ld{z_1}$ assigned in $\ld{t}$,  $\ld{l_p} ~\ld{\sqsubseteq}~ \ld{\lambda}(\ld{z_1})$ and 
for all $\ld{z_1}$ read in $\ld{t}$, $\ld{\lambda}(\ld{z_1}) ~\ld{\sqsubseteq}~ \ld{l_p}$.
Now if $\ld{l} \ld{~\sqsubseteq}~ \ld{l_p}$, then since in $\ld{t}$ no variables $\ld{z_2}$ such that
$\ld{\lambda}(\ld{z_2}) ~\ld{\sqsubseteq}~\ld{l}$ are assigned to.  
Therefore  by Lemma \ref{lemma:notassigned}, $\ld{\mu_f}(\ld{w}) =  \ld{\mu_1}(\ld{w}) = \ld{\nu_1}(\ld{w}) = \ld{\nu_f}(\ld{w})$, for all $\ld{w}$ such that $\ld{\lambda}(\ld{w}) ~\ld{\sqsubseteq}~ \ld{l}$. \\
If $\ld{l_p} \ld{~\sqsubseteq}~ \ld{l}$, then for all $\ld{z_1}$ read in $\ld{t}$, $\ld{\lambda}(\ld{z_1}) ~\ld{\sqsubseteq}~ \ld{l_p}$.  
Therefore, by assumption on transaction $\ld{t}$, if any variable $\ld{z}$  is assigned an expression $\ld{e}$,
since $\ld{\mu_1}$, $\ld{\nu_1}$ are two stores such that
for all $\ld{z_1} \in \ld{Z_e} = vars(\ld{e})$,  $\ld{\mu_1}(\ld{z_1}) = \ld{\nu_1}(\ld{z_1})$, the value of $\ld{e}$ will be the same. 
By this simple security argument, after the transaction $\ld{t}$, we have $\ld{\mu_f}(\ld{z}) = \ld{\nu_f}(\ld{z})$.
Since the transaction happened entirely and atomically in domain $\ld{L}$, we do not have to worry ourselves with changes in the other domain
$\rid{M}$, and do not need to concern ourselves with
the adversarial level $\rd{m}$.\\[1ex]
\textbf{Case} $p$ is $\ld{rd}(\ld{z},\ld{y})$: \ \ \ 
Thus $\langle \ld{\lambda}, \rd{\lambda} \rangle \vdash \ld{rd}(\ld{z},\ld{y}): 
\langle \ld{\lambda}(\ld{z}), \rd{m} \rangle$,
which means $\ld{\lambda}(\ld{y}) ~\ld{\sqsubseteq}~ \ld{\lambda}(\ld{z})$.
If $\ld{l} ~\ld{\sqsubseteq}~ \ld{\lambda}(\ld{z})$,
there is nothing to prove (Lemma \ref{lemma:notassigned}, again).
If $\ld{\lambda}(\ld{z}) ~\ld{\sqsubseteq}~ \ld{l}$, then since by I.H., $\ld{\mu_1}(\ld{y}) = \ld{\nu_1}(\ld{y})$, we have $\ld{\mu_f}(\ld{z}) =
\ld{\mu_1}[\ld{z} := \ld{\mu_1}(\ld{y})](\ld{z}) =
\ld{\nu_1}[\ld{z} := \ld{\nu_1}(\ld{y})](\ld{z}) =
\ld{\nu_f}(\ld{z})$. \\[1ex]
\textbf{Case} $p$ is $\ld{wr}(\ld{x},\ld{z})$: \ \ \ 
Thus $\langle \ld{\lambda}, \rd{\lambda} \rangle \vdash \ld{wr}(\ld{x},\ld{z}): 
\langle \ld{\lambda}(\ld{x}), \rd{m} \rangle$,
which means $\ld{\lambda}(\ld{z}) ~\ld{\sqsubseteq}~ \ld{\lambda}(\ld{x})$.
If $\ld{l} ~\ld{\sqsubseteq}~ \ld{\lambda}(\ld{x})$,
there is nothing to prove (Lemma \ref{lemma:notassigned}, again).
If $\ld{\lambda}(\ld{x}) ~\ld{\sqsubseteq}~ \ld{l}$, then since by I.H., $\ld{\mu_1}(\ld{z}) = \ld{\nu_1}(\ld{z})$, we have $\ld{\mu_f}(\ld{x}) =
\ld{\mu_1}[\ld{x} := \ld{\mu_1}(\ld{z})](\ld{x}) =
\ld{\nu_1}[\ld{x} := \ld{\nu_1}(\ld{z})](\ld{x}) = 
\ld{\nu_f}(\ld{x})$. \\[1ex]
\textbf{Case} $p$ is $T_{RL}(\ld{y},\rd{x})$: \ \ \ 
So $\langle \ld{\lambda}, \rd{\lambda} \rangle \vdash T_{RL}(\ld{y},\rd{x}): 
\langle \ld{\lambda}(\ld{y}), \rd{\lambda}(\rd{x}) \rangle$, and
$\ld{\gamma}(\rd{\lambda}(\rd{x})) ~\ld{\sqsubseteq}~ \ld{\lambda}(\ld{y})$.
If $\ld{l} ~\ld{\sqsubseteq}~ \ld{\lambda}(\ld{y})$,
there is nothing to prove (Lemma \ref{lemma:notassigned}, again).
If $\ld{\lambda}(\ld{y}) ~\ld{\sqsubseteq}~\ld{l}$,
then by transitivity, $\ld{\gamma}(\rd{\lambda}(\rd{x})) ~\ld{\sqsubseteq}~ \ld{l}$.
By monotonicity of $\rid{\alpha}$: 
$\rid{\alpha}(\ld{\gamma}(\rd{\lambda}(\rd{x}))) ~\rd{\sqsubseteq}~ \rid{\alpha}(\ld{l}) ~=~ \rd{m}$
(By our assumption on $\ld{l}$ and $\rd{m}$).
But by \textbf{LC2}, $\rd{\lambda}(\rd{x}) ~\rd{\sqsubseteq}~ \rid{\alpha}(\ld{\gamma}(\rd{\lambda}(\rd{x})))$.
So by transitivity, $\rd{\lambda}(\rd{x}) ~\rd{\sqsubseteq}~ \rd{m}$.
Now, by I.H., since $\rd{\mu_1}(\rd{x}) = \rd{\nu_1}(\rd{x})$, we have
$\ld{\mu_f}(\ld{y}) =
\ld{\mu_1}[\ld{y} := \rd{\mu_1}(\rd{x})](\ld{y}) =
\ld{\nu_1}[\ld{y} := \rd{\nu_1}(\rd{x})](\ld{y}) =
\ld{\nu_f}(\ld{y})$.
\end{proof}

\section{Finding Lagois Connections}\label{sec:revisit-lagois}

Lagois connections go well beyond providing a simple and elegant framework for secure connections between independent security lattices.  
They exhibit several properties that support the creation of secure connections, negotiating secure MoUs, and maintaining secure connections when organisational changes occur in the two connected lattices. 

Let $\rid{\alpha}[\ld{L}]$ and $\ld{\gamma}[\rid{M}]$ refer to the images of the order-preserving functions $\rid{\alpha}$ and $\ld{\gamma}$, respectively. 
Further, let us define the upward closures $\uparrow \rid{m} = \{ \rid{z} \in \rid{M} | \rid{m}\  \rd{\sqsubseteq} \rid{z} \}$ and $\uparrow \ld{l} = \{ \ld{z} \in \ld{L} | \ld{l}\  \ld{\sqsubseteq} \ld{z} \}$.

The two monotone functions 
$\rid{\alpha}: \ld{L} \rightarrow \rid{M}$ and 
$\ld{\gamma}: \rid{M} \rightarrow \ld{L}$ in a Lagois connection
$(\ld{L},\rid{\alpha},\ld{\gamma}, \rid{M})$
uniquely determine each other. 
\begin{proposition}[Proposition 3.9 in \cite{MELTON1994lagoisconnections}]\label{prop:fnguniquelydetermin}
If $(\ld{L},\rid{\alpha},\ld{\gamma}, \rid{M})$ is a Lagois connection, then the functions $\rid{\alpha}$ and $\ld{\gamma}$ uniquely determine each other, in fact:
\begin{align}
\ld{\gamma}(\rid{m}) = \bigsqcup \rid{\alpha}^{-1}[\bigsqcap \{ \rid{m^*} \in \rid{\alpha}[\ld{L}] \ | \rid{m} \rd{\sqsubseteq} \rid{m^*}\}] \\
= \bigsqcup \rid{\alpha}^{-1}[\bigsqcap ( \uparrow\rid{m} \cap \rid{\alpha}[\ld{L}])]
\end{align}
and
\begin{align}
    \rid{\alpha}(\ld{l}) = \bigsqcup \ld{\gamma}^{-1}[\bigsqcap \{ \ld{l^*}\in \ld{\gamma}[\rid{M}] \ | \ld{l} \ld{\sqsubseteq} \ld{l^*}\}] \\
= \bigsqcup \ld{\gamma}^{-1}[\bigsqcap ( \uparrow\ld{l} \cap \ld{\gamma}[\rid{M}])]
\end{align}
\end{proposition}

\subsection{Negotiating an MoU when given one order-preserving map}\label{sec:lagois-adjoint}

Suppose we are given two security lattices $\ld{L}$ and $\rid{M}$ and an order-preserving function $\rid{\alpha}: \ld{L} \rightarrow \rid{M}$, we can \textit{find a Lagois adjoint} $\ld{\gamma}: \rid{M} \rightarrow \ld{L}$, i.e., an order-preserving function such that $(\ld{L},\rid{\alpha},\ld{\gamma}, \rid{M})$ is a Lagois connection, and thus secure information flow between $\ld{L}$ and $\rid{M}$ is ensured. 
\begin{proposition}[Proposition 3.10 in \cite{MELTON1994lagoisconnections}]\label{prop:existlagoisadjoint}
Let $\ld{L}$ and $\rid{M}$ be posets. Then an order-preserving function $\rid{\alpha}: \ld{L} \rightarrow \rid{M}$ has a Lagois adjoint $\ld{\gamma}: \rid{M} \rightarrow \ld{L}$ iff:
\begin{enumerate}
    \item $\rid{\alpha}^{-1}(\rid{m})$ has a largest member, for all $\rid{m} \in \rid{\alpha}[\ld{L}]$.
    \item $\uparrow \rid{m} \cap \rid{\alpha}[\ld{L}]$ has a smallest member, for all $\rid{m} \in \rid{M}$.
    \item The restriction of $\rid{\alpha}$ from $\{ \bigsqcup \rid{\alpha}^{-1}(\rid{m}) | \rid{m} \in \rid{\alpha}[\ld{L}]\}$ to its image $\rid{\alpha}[\ld{L}]$ is an order isomorphism. 
\end{enumerate}
\end{proposition}

Consider the example of the order-preserving function $\rid{\alpha}$ in Figure \ref{fig:WnD}. 
First we check whether the given function  $\rid{\alpha}$ has a Lagois adjoint or not, using Proposition \ref{prop:existlagoisadjoint}. 
The conditions on $\rid{\alpha}$ are obviously satisfied in our example. 
We then use the third condition of Proposition \ref{prop:existlagoisadjoint} to identify the order-isomorphic substructures of the participating security lattices. 
The security classes which form such an order-isomorphic structure for our example are shown in blue/red in Figure \ref{fig:findingLCfromalpha}.

We allow the information to flow back to domain $\ld{L}$  by mapping first only the \textit{budpoints} in $\rid{M}$ \textit{to the budpoints} in $\ld{L}$, while respecting the conditions for Lagois connections, i.e., $\mathbf{LC3}$ and $\mathbf{LC4}$. 
These mappings are shown as dotted brown arrows in Figure \ref{fig:findingLCfromalpha}.
Then we invoke Proposition \ref{prop:fnguniquelydetermin} 
to complete the mappings of function $\ld{\gamma}$. 
The unmapped security classes in domain $\rid{M}$, say $\rid{m_1}$, are connected to the greatest of those elements of $\ld{L}$ which are mapped by $\rid{\alpha}$ to the least budpoint in $\rid{M}$ greater than $\rid{m_1}$, as shown by solid brown arrows in Figure \ref{fig:findingLCfromalpha2}.

\begin{figure}[!ht]
\centering
    \begin{minipage}[t]{.45\textwidth}
    \begin{tikzpicture}[framed,->,node distance=1cm,on grid]
    \title{W and D}
    \node(T2)   {{\color{red}\scriptsize$\top 2$}};
    \node(T1)  [xshift=-3cm]  {{\color{blue} \scriptsize$\top 1$}};
    \node(Dir1) [below of = T1] {{\color{blue}\scriptsize$CollegePrincipal$}};
    \node(D1) [below left of = Dir1] {\scriptsize$Dean\ (F)$};
    \node(F1) [below of = D1] {\scriptsize$Faculty$};
    \node(DS1) [xshift=-2.3cm,yshift=-1.7cm] {\scriptsize$Dean\ (S)$};
    \node(S1) [below right of = F1] {{\color{blue}\scriptsize$Student$}};
    \node(B1)  [below of=S1]  {{\color{blue}\scriptsize$\bot 1$}};
    \node(Sec2)  [below of = T2] {\scriptsize$Chancellor$};
    \node[align=center] (AS2)  [below of=Sec2] {\scriptsize $Vice$ \\ \scriptsize $Chancellor$};
    \node(Dir2)  [below of=AS2] {{\color{red}\scriptsize$Dean(Colleges)$}};
    \node(E2)  [below of=Dir2] {{\color{red}\scriptsize$Univ.Fac.$}};
    \node(B2)  [below of=E2]  {{\color{red}\scriptsize$\bot 2$}};
    \draw [OliveGreen, thick] (F1) to (E2);
    \draw [OliveGreen, ->, thick] (S1) to (E2);
    \draw [OliveGreen, thick] (B1) to (B2);
    \draw [OliveGreen, thick] (T1) to (T2);
    \draw [OliveGreen, thick] (Dir1) [bend left = 10] to (Dir2);
    \draw [OliveGreen, thick] (D1) to (Dir2);
    \draw [OliveGreen, thick] (DS1) to (Dir2);
    \draw [Fuchsia, ->] (Dir1) to (T1);
    \draw [Fuchsia, ->] (D1) to (Dir1);
    \draw [Fuchsia, ->] (F1) to (D1);
    \draw [Fuchsia, ->] (S1) to (F1);
    \draw [Fuchsia, ->] (S1) to (DS1);
    \draw [Fuchsia, ->] (DS1) to (Dir1);
    \draw [Fuchsia, ->] (B1) to (S1);
    \draw [Fuchsia, ->] (B2) to (E2);
    \draw [Fuchsia, ->] (E2) to (Dir2);
    \draw [Fuchsia, ->] (Dir2) to (AS2);
    \draw [Fuchsia, ->] (AS2) to (Sec2);
    \draw [Fuchsia, ->] (Sec2) to (T2);
    
    \draw [brown,  dotted, thick, <-] (F1)[bend left = 10] to (E2);
    \draw [brown,  dotted, thick,<-] (B1)[bend left = 15] to (B2);
    \draw [brown,  dotted, thick, <-] (T1) [bend right = 15] to (T2);
    \draw [brown,  dotted, thick, <-] (Dir1) [bend left = 20] to (Dir2);
    
    \draw (-3,-2.5)[blue] ellipse (1.4cm and 2.7cm);
    \draw (0,-2.5)[red] ellipse (1.2cm and 2.7cm);
    
    \draw [dotted] (0,-1.4) ellipse (0.8cm and 1.2cm);
    \draw [dotted] (-3,-1.7) ellipse (1.35cm and 0.3cm);
    \draw [dotted] (-3,-3.2) ellipse (0.5cm and 0.55cm);
    
    \end{tikzpicture}
    \caption{\small Connecting budpoints while finding a viable Lagois \textit{adjoint} for a given order-preserving function $\rid{\alpha}$ (from Figure \ref{fig:WnD}). } \label{fig:findingLCfromalpha}
    \end{minipage}
    \quad \quad
    \begin{minipage}[t]{.45\textwidth}
    \begin{tikzpicture}[framed,->,node distance=1cm,on grid]
    \title{W and D}
    \node(T2)   {{\color{red}\scriptsize$\top 2$}};
    \node(T1)  [xshift=-3cm]  {{\color{blue} \scriptsize$\top 1$}};
    \node(Dir1) [below of = T1] {{\color{blue}\scriptsize$CollegePrincipal$}};
    \node(D1) [below left of = Dir1] {\scriptsize$Dean\ (F)$};
    \node(F1) [below of = D1] {\scriptsize$Faculty$};
    \node(DS1) [xshift=-2.3cm,yshift=-1.7cm] {\scriptsize$Dean\ (S)$};
    \node(S1) [below right of = F1] {{\color{blue}\scriptsize$Student$}};
    \node(B1)  [below of=S1]  {{\color{blue}\scriptsize$\bot 1$}};
    \node(Sec2)  [below of = T2] {\scriptsize$Chancellor$};
    \node[align=center] (AS2)  [below of=Sec2] {\scriptsize $Vice$ \\ \scriptsize $Chancellor$};
    \node(Dir2)  [below of=AS2] {{\color{red}\scriptsize$Dean(Colleges)$}};
    \node(E2)  [below of=Dir2] {{\color{red}\scriptsize$Univ.Fac.$}};
    \node(B2)  [below of=E2]  {{\color{red}\scriptsize$\bot 2$}};
    \draw [densely dashed, black,  thick, <->] (F1) to (E2);
    \draw [OliveGreen, ->, thick] (S1) to (E2);
    \draw [densely dashed, black,  thick,<->] (B1) to (B2);
    \draw [densely dashed, black,  thick, <->] (T1) to (T2);
    \draw [densely dashed, black,  thick, <->] (Dir1) [bend left = 10] to (Dir2);
    \draw [OliveGreen, thick] (D1) to (Dir2);
    \draw [OliveGreen, thick] (DS1) to (Dir2);
    \draw [brown, thick] (AS2) to (T1);
    \draw [brown, thick] (Sec2) to (T1);
    \draw [Fuchsia, ->] (Dir1) to (T1);
    \draw [Fuchsia, ->] (D1) to (Dir1);
    \draw [Fuchsia, ->] (F1) to (D1);
    \draw [Fuchsia, ->] (S1) to (F1);
    \draw [Fuchsia, ->] (S1) to (DS1);
    \draw [Fuchsia, ->] (DS1) to (Dir1);
    \draw [Fuchsia, ->] (B1) to (S1);
    \draw [Fuchsia, ->] (B2) to (E2);
    \draw [Fuchsia, ->] (E2) to (Dir2);
    \draw [Fuchsia, ->] (Dir2) to (AS2);
    \draw [Fuchsia, ->] (AS2) to (Sec2);
    \draw [Fuchsia, ->] (Sec2) to (T2);
    
    \draw (-3,-2.5)[blue] ellipse (1.4cm and 2.7cm);
    \draw (0,-2.5)[red] ellipse (1.2cm and 2.7cm);
    
    \draw [dotted] (0,-1.4) ellipse (0.8cm and 1.2cm);
    \draw [dotted] (-3,-1.7) ellipse (1.35cm and 0.3cm);
    \draw [dotted] (-3,-3.2) ellipse (0.5cm and 0.55cm);
    
    \end{tikzpicture}
    \caption{\small Defining a viable \textit{Lagois adjoint} for a given $\rid{\alpha}$ (in Figure \ref{fig:WnD})  
    \label{fig:findingLCfromalpha2}} 
    \end{minipage}
\end{figure}

\subsection{Negotiating an MoU \textit{ab initio}} \label{sec:lagois-abinitio}

In the absence of any constraints on $\rid{\alpha}$ and $\ld{\gamma}$, it is always possible to define a Lagois connection between two finite security class lattices, e.g., by mapping all elements to the topmost security class of the other lattice. 
But such a Lagois connection is of little use to the participating organisations as the shared information, being in the topmost security class, is inaccessible to most principals in the respective organisations.  An operative result for creating viable\footnote{We use term ``viable" informally to mean that the natural constraints of the application are taken into account, and that the security level of information is escalated only to the extent required, thus not making its access overly restricted.} secure MoUs is Theorem \ref{lemma:existenceLC}.
\begin{theorem}[Theorem 3.20 in \cite{MELTON1994lagoisconnections}] \label{lemma:existenceLC} Let $(\ld{L},\ld{\sqsubseteq})$ and $(\rid{M},\rd{\sqsubseteq})$ be posets. There is a Lagois connection between $(\ld{L},\ld{\sqsubseteq})$ and $(\rid{M},\rd{\sqsubseteq})$ if and only if the following four conditions hold:
 \begin{enumerate} 
    \item There exist order-isomorphic subsets $\ld{L^*} \subseteq \ld{L}$ and $\rid{{M^*}} \subseteq \rid{M}$.
    \item There exists equivalence relations $\ld{\sim_L}$ on $\ld{L}$ and $\rd{\sim_M}$ on $\rid{M}$ 
    such that $\ld{L^*}$ is a system of representatives for $\ld{\sim_L}$ and $\rid{{M^*}}$ is a system of representatives for $\rd{\sim_M}$ respectively\footnote{The members of $\ld{L^*}$ and $\rid{{M^*}}$ are called budpoints and the equivalence classes are called blossoms.}.
    \item if $\ld{l_1} \in \ld{L}$ and $\ \ld{l^*} \in \ld{L^*}$ with $\ld{l_1} \ld{\sim_L} \ld{l^*}$, then $\ld{l_1} \ld{\sqsubseteq} \ld{l^*}$; and if $\rid{m_1} \in \rid{M}$ and $\rid{m^*} \in \rid{M^*}$ with $\rid{m_1} \rd{\sim _M}\ \rid{m^*}$, then $\rid{m_1} \rd{\sqsubseteq} \rid{m^*}$.
    \item If $\ld{l_1} \ld{\sqsubseteq} \ld{l_2}$ in $\ld{L}$ and $\ld{l^*_1}, \ld{l^*_2} \in \ld{L^*}$ with $\ld{l_1} \ld{\sim_L} \ld{l^*_1}$ and $\ld{l_2} \ld{\sim_L} \ld{l^*_2}$, then $\ld{l^*_1} \ld{\sqsubseteq} \ld{l^*_2}$; and if $\rid{m_1} \rd{\sqsubseteq} \rid{m_2}$ and $\rid{m^*_1}, \rid{m^*_2} \in \rid{M^*}$ with $\rid{m_1} \rd{\sim_M} \rid{m^*_1}$ and $\rid{m_2} \rd{\sim_M} \rid{m^*_2}$, then $\rid{m^*_1} \rd{\sqsubseteq} \rid{m^*_2}$.
    \end{enumerate}
\end{theorem}
\begin{corollary}[Corollary 3.21 in \cite{MELTON1994lagoisconnections}] \label{cor:existenceLC2}
    Let $\ld{L}$ and $\rid{M}$ be posets, and $\ld{c}: \ld{L} \rightarrow \ld{L}$ and $\rid{i}: \rid{M} \rightarrow \rid{M}$ be closure operators such that $\ld{c}[\ld{L}]$ and $\rid{i}[\rid{M}]$ are isomorphic (with their inherited orders). If $\rid{h}: \ld{c}[\ld{L}] \rightarrow \rid{i}[\rid{M}]$ is such an isomorphism, then $(\ld{L}, \rid{h}\ld{c}, \ld{h^{-1}}\rid{i}, \rid{M})$ is a Lagois connection.
    \end{corollary}
Theorem \ref{lemma:existenceLC} suggests the following method, which we illustrate using an example in 
Figures \ref{fig:LCfrmScratchA}-\ref{fig:LCfrmScratchWidClosures}, 
where an agreement is negotiated between \textit{Dorm-Life} and \textit{College}.
\begin{enumerate}
    \item Find the maximal order-isomorphic substructures $\ld{L^*} \subseteq \ld{L}$ and $\rid{M^*} \subseteq \rid{M}$ (which include the transfer classes
    [\ld{\textit{Student}}-\rid{\textit{Student}}, \ld{\textit{HouseMaster}}-\rid{\textit{Dean(S)}}], the bottom-most and the topmost class) in the given security class lattices $(\ld{L},\ld{\sqsubseteq})$ and $(\rid{M},\rd{\sqsubseteq})$. 
     These classes in the order-isomorphic structures of the two domains are  indicated in bold in Figure \ref{fig:LCfrmScratchA}.
    \item Identify an equivalence relation $\ld{\sim_L}$ such that $\ld{L^*}$ is a system of representatives for $\ld{\sim_L}$. Similarly, identify an equivalence relation $\rd{\sim_M}$ such that $\rid{M^*}$ is a system of representatives for $\rd{\sim _M}$. The equivalence relations should be such that the two conditions (3) and (4) of Theorem \ref{lemma:existenceLC} hold, which essentially ensure that one is allowed to reason about information flows from all the security classes in the given security class lattices.
 
    The members of $\ld{L^*}$ (resp. $\rid{M^*}$), called budpoints,
     play a significant role in delineating the connection between the transfer classes in the two lattices.
    We show the equivalence classes for our example in Figure \ref{fig:LCfrmScratchB}, where the budpoints are the security classes coloured blue and red in the respective domains.
    \item Now interconnect the budpoints of both domains to each other while respecting the inherited order, forming an isomorphic structure, as done in  Figure \ref{fig:LCfrmScratchC} for our example.
    \item 
    Now keeping in mind Corollary \ref{cor:existenceLC2},
    we define closure operators $\ld{c}: \ld{L} \rightarrow \ld{L}$ and $\rd{i}: \rid{M} \rightarrow \rid{M}$, using the equivalence relations $\ld{\sim_L}$ and $\rd{\sim _M}$, such that $\ld{c}[\ld{L}]$ and $\rd{i}[\rid{M}]$ are isomorphic (with their inherited orders).  Figure \ref{fig:LCfrmScratchWidClosures} shows the closure operators for our example.
    \item Thereafter, define an increasing Lagois connection using Corollary \ref{cor:existenceLC2} as $(\ld{L}, \rid{h}\ld{c}, \ld{h^{-1}}\rd{i}, \rid{M})$. For our example, the Lagois connection is defined as shown in Figure \ref{fig:LCfrmScratchD}, with green arrows completing the mapping from \textit{Dorm-Life} and brown arrows completing the mapping from \textit{College}.
\end{enumerate}

\begin{figure}[!ht]
    \begin{minipage}[t]{.45\textwidth}
    \begin{tikzpicture}[framed,->,node distance=1cm,on grid]
    \title{W and D}
    \node(T1)    { \scriptsize$\mathbf{\top 1}$};
    \node(Dir1) [below of = T1] {\scriptsize$CollegePrincipal$};
    \node(D1) [below left of = Dir1] {\scriptsize$Dean\ (F)$};
    \node(F1) [below of = D1] {\scriptsize$Faculty$};
    \node(DS1) [below right of = Dir1] {\scriptsize$\mathbf{Dean\ (S)}$};
    \node(S1) [below right of = F1] {\scriptsize$\mathbf{Student}$};
    \node(B1)  [below of=S1]  {\scriptsize$\mathbf{\bot 1}$};
    \draw [Fuchsia, ->] (Dir1) to (T1);
    \draw [Fuchsia, ->] (D1) to (Dir1);
    \draw [Fuchsia, ->] (F1) to (D1);
    \draw [Fuchsia, ->] (S1) to (F1);
    \draw [Fuchsia, ->] (S1) to (DS1);
    \draw [Fuchsia, ->] (DS1) to (Dir1);
    \draw [Fuchsia, ->] (B1) to (S1);
    
    \draw (0,-2.5)[red] ellipse (1.4cm and 2.7cm);
    

   \node(T1')  [xshift=-3cm]  { \scriptsize$\mathbf{\top 0}$};
   
    \node(Dir1') [below of = T1'] {\scriptsize$\mathbf{HouseMaster}$};
    \node(D1') [below left of = Dir1'] {\scriptsize$Caretaker$};
    \node[align=center](F1') [below of = D1'] {\scriptsize $Assistant$};
    \node[align=center](DS1') [below right of = Dir1'] {\scriptsize $Dining-$ \\ \scriptsize $Manager$};
    \node(S1') [below right of = F1'] {\scriptsize$\mathbf{Student}$};
    \node(B1')  [below of=S1']  {\scriptsize$\mathbf{\bot 0}$};
    \draw [Fuchsia, ->] (Dir1') to (T1');
    \draw [Fuchsia, ->] (D1') to (Dir1');
    \draw [Fuchsia, ->] (F1') to (D1');
    \draw [Fuchsia, ->] (S1') to (F1');
    \draw [Fuchsia, ->] (S1') to (DS1');
    \draw [Fuchsia, ->] (Dir1') to (DS1');
    \draw [Fuchsia, ->] (B1') to (S1');
    
    
    
    \draw (-3,-2.5)[blue] ellipse (1.4cm and 3.1cm);
    

    \end{tikzpicture}
    \caption{\small Security lattices for two autonomous organisations that want to negotiate a \textit{secure} MoU \textit{ab initio}.
    \label{fig:LCfrmScratchA}} 
    \end{minipage}
    \quad \quad 
    \begin{minipage}[t]{.45\textwidth}
        \begin{tikzpicture}[framed,->,node distance=1cm,on grid]
    \title{W and D}
    \node(T1)    {{\color{red} \scriptsize$\top 1$}};
    \node(Dir1) [below of = T1] {\scriptsize$CollegePrincipal$};
    \node(D1) [below left of = Dir1] {\scriptsize$Dean\ (F)$};
    \node(F1) [below of = D1] {\scriptsize$Faculty$};
    \node(DS1) [below right of = Dir1] {{\color{red}\scriptsize$Dean\ (S)$}};
    \node(S1) [below right of = F1] {{\color{red}\scriptsize$Student$}};
    \node(B1)  [below of=S1]  {{\color{red}\scriptsize$\bot 1$}};
    \draw [Fuchsia, ->] (Dir1) to (T1);
    \draw [Fuchsia, ->] (D1) to (Dir1);
    \draw [Fuchsia, ->] (F1) to (D1);
    \draw [Fuchsia, ->] (S1) to (F1);
    \draw [Fuchsia, ->] (S1) to (DS1);
    \draw [Fuchsia, ->] (DS1) to (Dir1);
    \draw [Fuchsia, ->] (B1) to (S1);
    
    \draw (0,-2.5)[red] ellipse (1.4cm and 2.7cm);
    

   \node(T1')  [xshift=-3cm]  {{\color{blue} \scriptsize$\top 0$}};
   
    \node(Dir1') [below of = T1'] {{\color{blue}\scriptsize$HouseMaster$}};
    \node(D1') [below left of = Dir1'] {\scriptsize$Caretaker$};
    \node[align=center](F1') [below of = D1'] {\scriptsize $Assistant$};
    \node[align=center](DS1') [below right of = Dir1'] {\scriptsize $Dining-$ \\ \scriptsize $Manager$};
    \node(S1') [below right of = F1'] {{\color{blue}\scriptsize$Student$}};
    \node(B1')  [below of=S1']  {{\color{blue}\scriptsize$\bot 0$}};
    \draw [Fuchsia, ->] (Dir1') to (T1');
    \draw [Fuchsia, ->] (D1') to (Dir1');
    \draw [Fuchsia, ->] (F1') to (D1');
    \draw [Fuchsia, ->] (S1') to (F1');
    \draw [Fuchsia, ->] (S1') to (DS1');
   \draw [Fuchsia, ->] (Dir1') to (DS1');
    \draw [Fuchsia, ->] (B1') to (S1');
    
    
    
    \draw (-3,-2.5) [blue] ellipse (1.4cm and 3.1cm);
    
    
    \draw [dotted,thick] (-3,-0) ellipse (0.3cm and 0.3cm);
    \draw [dotted,thick] (-3,-4.5) ellipse (0.3cm and 0.3cm);
    \draw [rotate around={-15:(-0.4,-1.4)},dotted,thick] (-0.4,-1.4) ellipse (0.5cm and 1.7cm);
    
    \draw [dotted,thick] (-0,-4.5) ellipse (0.3cm and 0.3cm);
    \draw [dotted,thick] (-3,-3.4) ellipse (0.6cm and 0.3cm);
    \draw [dotted,thick] (-0,-3.4) ellipse (0.6cm and 0.3cm);
    \draw [dotted,thick] (-3,-1.9) ellipse (1.5cm and 1.1cm);
    \draw [dotted,thick] (0.7,-1.7) ellipse (0.6cm and 0.3cm);

    \end{tikzpicture}
    \caption{\small Identifying equivalence relations in given security lattices for discovering a Lagois connection \textit{ab initio}.
    \label{fig:LCfrmScratchB}} 
    \end{minipage}
\end{figure}

\begin{figure}[!ht]
    \begin{minipage}[t]{.45\textwidth}
    \begin{tikzpicture}[framed,->,node distance=1cm,on grid]
    \title{W and D}
    \node(T1)    {{\color{red} \scriptsize$\top 1$}};
    \node(Dir1) [below of = T1] {\scriptsize$CollegePrincipal$};
    \node(D1) [below left of = Dir1] {\scriptsize$Dean\ (F)$};
    \node(F1) [below of = D1] {\scriptsize$Faculty$};
    \node(DS1) [below right of = Dir1] {{\color{red}\scriptsize$Dean\ (S)$}};
    \node(S1) [below right of = F1] {{\color{red}\scriptsize$Student$}};
    \node(B1)  [below of=S1]  {{\color{red}\scriptsize$\bot 1$}};
    \draw [Fuchsia, ->] (Dir1) to (T1);
    \draw [Fuchsia, ->] (D1) to (Dir1);
    \draw [Fuchsia, ->] (F1) to (D1);
    \draw [Fuchsia, ->] (S1) to (F1);
    \draw [Fuchsia, ->] (S1) to (DS1);
    \draw [Fuchsia, ->] (DS1) to (Dir1);
    \draw [Fuchsia, ->] (B1) to (S1);
    \draw (0,-2.5)[red] ellipse (1.4cm and 2.7cm);
   \node(T1')  [xshift=-3cm]  {{\color{blue} \scriptsize$\top 0$}};
    \node(Dir1') [below of = T1'] {{\color{blue}\scriptsize$HouseMaster$}};
    \node(D1') [below left of = Dir1'] {\scriptsize$Caretaker$};
    \node[align=center](F1') [below of = D1'] {\scriptsize $Assistant$};
    \node[align=center](DS1') [below right of = Dir1'] {\scriptsize $Dining-$ \\ \scriptsize $Manager$};
    \node(S1') [below right of = F1'] {{\color{blue}\scriptsize$Student$}};
    \node(B1')  [below of=S1']  {{\color{blue}\scriptsize$\bot 0$}};
    \draw [Fuchsia, ->] (Dir1') to (T1');
    \draw [Fuchsia, ->] (D1') to (Dir1');
    \draw [Fuchsia, ->] (F1') to (D1');
    \draw [Fuchsia, ->] (S1') to (F1');
    \draw [Fuchsia, ->] (S1') to (DS1');
    \draw [Fuchsia, ->] (Dir1') to (DS1');
    \draw [Fuchsia, ->] (B1') to (S1');
    \draw [densely dashed, black,  thick, <->] (T1') to (T1);
    \draw [densely dashed, black,  thick, <->] (B1') to (B1);
    \draw [densely dashed, black,  thick, <->] (S1') to (S1);
    \draw [densely dashed, black,  thick, <->] (Dir1') to [bend left = 5] (DS1);
    \draw (-3,-2.5)[blue] ellipse (1.4cm and 3.1cm);
    \draw [dotted,thick] (-3,-0) ellipse (0.3cm and 0.3cm);
    \draw [dotted,thick] (-3,-4.5) ellipse (0.3cm and 0.3cm);
    \draw [rotate around={-15:(-0.4,-1.4)},dotted,thick] (-0.4,-1.4) ellipse (0.5cm and 1.7cm);
    \draw [dotted,thick] (-0,-4.5) ellipse (0.3cm and 0.3cm);
    \draw [dotted,thick] (-3,-3.4) ellipse (0.6cm and 0.3cm);
    \draw [dotted,thick] (-0,-3.4) ellipse (0.6cm and 0.3cm);
    \draw [dotted,thick] (-3,-1.9) ellipse (1.5cm and 1.1cm);
    \draw [dotted,thick] (0.7,-1.7) ellipse (0.6cm and 0.3cm);
    \end{tikzpicture}
    
    \caption{\small Connecting budpoints of equivalent security classes.
    \label{fig:LCfrmScratchC}}
    \end{minipage}
    \quad \quad 
    \begin{minipage}[t]{.45\textwidth}
        \begin{tikzpicture}[framed,->,node distance=1cm,on grid]
    \title{W and D}
    \node(T1)    {{\color{red} \scriptsize$\top 1$}};
    \node(Dir1) [below of = T1] {\scriptsize$CollegePrincipal$};
    \node(D1) [below left of = Dir1] {\scriptsize$Dean\ (F)$};
    \node(F1) [below of = D1] {\scriptsize$Faculty$};
    \node(DS1) [below right of = Dir1] {{\color{red}\scriptsize$Dean\ (S)$}};
    \node(S1) [below right of = F1] {{\color{red}\scriptsize$Student$}};
    \node(B1)  [below of=S1]  {{\color{red}\scriptsize$\bot 1$}};
    \draw [Fuchsia, ->] (Dir1) to (T1);
    \draw [Fuchsia, ->] (D1) to (Dir1);
    \draw [Fuchsia, ->] (F1) to (D1);
    \draw [Fuchsia, ->] (S1) to (F1);
    \draw [Fuchsia, ->] (S1) to (DS1);
    \draw [Fuchsia, ->] (DS1) to (Dir1);
    \draw [Fuchsia, ->] (B1) to (S1);
    \draw (0,-2.5)[red] ellipse (1.4cm and 2.7cm);
   \node(T1')  [xshift=-3cm]  {{\color{blue} \scriptsize$\top 0$}};
    \node(Dir1') [below of = T1'] {{\color{blue}\scriptsize$HouseMaster$}};
    \node(D1') [below left of = Dir1'] {\scriptsize$Caretaker$};
    \node[align=center](F1') [below of = D1'] {\scriptsize $Assistant$};
    \node[align=center](DS1') [below right of = Dir1'] {\scriptsize $Dining-$ \\ \scriptsize $Manager$};
    \node(S1') [below right of = F1'] {{\color{blue}\scriptsize$Student$}};
    \node(B1')  [below of=S1']  {{\color{blue}\scriptsize$\bot 0$}};
    \draw [brown,thick, ->] (F1) to [bend right = 10] (T1');
    \draw [brown,thick, ->] (D1) to [bend right = 10](T1');
    \draw [brown,thick, ->] (Dir1) to (T1');
     \draw [OliveGreen,thick, ->] (D1') to [bend right = 20](DS1);
    \draw [OliveGreen,thick, ->] (F1') to [bend right = 15] (DS1);
    \draw [OliveGreen,thick, ->] (DS1') to [bend right = 15] (DS1);
    
    \draw [Fuchsia, ->] (Dir1') to (T1');
    \draw [Fuchsia, ->] (D1') to (Dir1');
    \draw [Fuchsia, ->] (F1') to (D1');
    \draw [Fuchsia, ->] (S1') to (F1');
    \draw [Fuchsia, ->] (S1') to (DS1');
    \draw [Fuchsia, ->] (Dir1') to (DS1');
    \draw [Fuchsia, ->] (B1') to (S1');
    \draw [densely dashed, black,  thick, <->] (T1') to (T1);
    \draw [densely dashed, black,  thick, <->] (B1') to (B1);
    \draw [densely dashed, black,  thick, <->] (S1') to (S1);
    \draw [densely dashed, black,  thick, <->] (Dir1') to [bend left = 5] (DS1);
    \draw (-3,-2.5) [blue] ellipse (1.4cm and 3.1cm);
    \end{tikzpicture}
    \caption{\small A secure MoU negotiated \textit{ab initio}.
    \label{fig:LCfrmScratchD}} 
    \end{minipage}
\end{figure}

\begin{figure}[!ht]
    \centering
    \begin{tikzpicture}[framed,->,node distance=1cm,on grid]
    \title{W and D}
    \node(T1)  [xshift=3cm]  {{\color{red} \scriptsize$\top 1$}};
    \node(Dir1) [below of = T1] {\scriptsize$CollegePrincipal$};
    \node(D1) [below left of = Dir1] {\scriptsize$Dean\ (F)$};
    \node(F1) [below of = D1] {\scriptsize$Faculty$};
    \node(DS1) [below right of = Dir1] {{\color{red}\scriptsize$Dean\ (S)$}};
    \node(S1) [below right of = F1] {{\color{red}\scriptsize$Student$}};
    \node(B1)  [below of=S1]  {{\color{red}\scriptsize$\bot 1$}};
    
    \node(T12) [yshift=-1cm] {{\color{red} \scriptsize$\top 1$}};
    \node(DS12) [below  of = T12] {{\color{red}\scriptsize$Dean\ (S)$}};
    \node(S12) [below  of = DS12] {{\color{red}\scriptsize$Student$}};
    \node(B12)  [below of=S12]  {{\color{red}\scriptsize$\bot 1$}};
    
    \draw [Fuchsia, ->] (Dir1) to (T1);
    \draw [Fuchsia, ->] (D1) to (Dir1);
    \draw [Fuchsia, ->] (F1) to (D1);
    \draw [Fuchsia, ->] (S1) to (F1);
    \draw [Fuchsia, ->] (S1) to (DS1);
    \draw [Fuchsia, ->] (DS1) to (Dir1);
    \draw [Fuchsia, ->] (B1) to (S1);
    \draw (3,-2.5)[red] ellipse (1.4cm and 2.7cm);
    \draw (0,-2.5)[red] ellipse (1.4cm and 2.7cm);
    
   \node(T1')  [xshift=-6cm]  {{\color{blue} \scriptsize$\top 0$}};
    \node(Dir1') [below of = T1'] {{\color{blue}\scriptsize$HouseMaster$}};
    \node(D1') [below left of = Dir1'] {\scriptsize$Caretaker$};
    \node[align=center](F1') [below of = D1'] {\scriptsize $Assistant$};
    \node[align=center](DS1') [below right of = Dir1'] {\scriptsize $Dining-$ \\ \scriptsize $Manager$};
    \node(S1') [below right of = F1'] {{\color{blue}\scriptsize$Student$}};
    \node(B1')  [below of=S1']  {{\color{blue}\scriptsize$\bot 0$}};
    
    \node(T12')  [xshift=-3cm,yshift=-1cm]  {{\color{blue} \scriptsize$\top 0$}};
    \node(Dir12') [below of = T12'] {{\color{blue}\scriptsize$HouseMaster$}};
    \node(S12') [below of = Dir12'] {{\color{blue}\scriptsize$Student$}};
    \node(B12')  [below of=S12']  {{\color{blue}\scriptsize$\bot 0$}};
    
    \draw [purple,thick, ->] (T1) to (T12);
    \draw [purple,thick, ->] (B1) to (B12);
    \draw [purple,thick, ->] (S1) to (S12);
    \draw [purple,thick, ->] (DS1) to [bend left = 20] (DS12);
    \draw [purple,thick, ->] (T1') to (T12');
    \draw [purple,thick, ->] (B1') to (B12');
    \draw [purple,thick, ->] (S1') to (S12');
    \draw [purple,thick, ->] (F1) to [bend right = 10] (T12);
    \draw [purple,thick, ->] (D1) to [bend right = 10](T12);
    \draw [purple,thick, ->] (Dir1) to (T12);
    \draw [purple,thick, ->] (T1) to (T12);
     \draw [purple,thick, ->] (D1') to [bend right = 20](Dir12');
    \draw [purple,thick, ->] (F1') to [bend right = 15] (Dir12');
    \draw [purple,thick, ->] (DS1') to [bend right = 15] (Dir12');
    \draw [purple,thick, ->] (Dir1') to (Dir12');
    
    \draw [Fuchsia, ->] (B12') to (S12');
    \draw [Fuchsia, ->] (S12') to (Dir12');
    \draw [Fuchsia, ->] (Dir12') to (T12');
    
    \draw [Fuchsia, ->] (B12) to (S12);
    \draw [Fuchsia, ->] (S12) to (DS12);
    \draw [Fuchsia, ->] (DS12) to (T12);
    
    \draw [Fuchsia, ->] (Dir1') to (T1');
    \draw [Fuchsia, ->] (D1') to (Dir1');
    \draw [Fuchsia, ->] (F1') to (D1');
    \draw [Fuchsia, ->] (S1') to (F1');
    \draw [Fuchsia, ->] (S1') to (DS1');
    \draw [Fuchsia, ->] (Dir1') to (DS1');
    \draw [Fuchsia, ->] (B1') to (S1');
    \draw [densely dashed, black,  thick, <->] (T12') to (T12);
    \draw [densely dashed, black,  thick, <->] (B12') to (B12);
    \draw [densely dashed, black,  thick, <->] (S12') to (S12);
    \draw [densely dashed, black,  thick, <->] (Dir12') to (DS12);
    \draw (-6,-2.5) [blue] ellipse (1.4cm and 3.1cm);
    \draw (-3,-2.5) [blue] ellipse (1.4cm and 3.1cm);
    \end{tikzpicture}
    \caption{\small Using isomorphic images of closure operators to define a  Lagois connection. Purple edges define the closure operators for each organisational domain.
    \label{fig:LCfrmScratchWidClosures}} 
\end{figure}
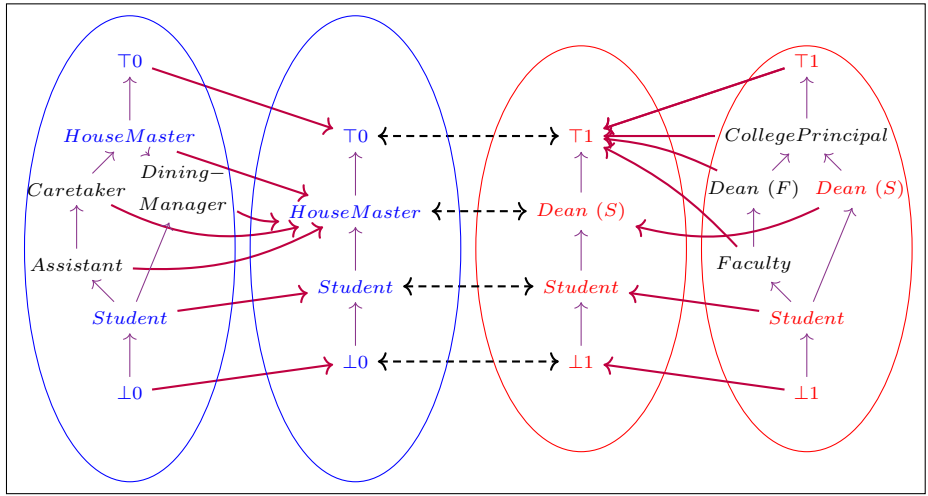

\subsection{MoUs involving several administrative domains}\label{sec:lagois-composition}

Suppose there is a sequence of administrative domains such that each adjacent pair of domains has negotiated a secure Lagois connection to ensure bidirectional SIF between them.
Theorem \ref{lemma:composinglagoisconnections} and Corollary \ref{lemma:composition2} allow the composition of these  Lagois connections using simple functional composition.

\begin{theorem} [Theorem 3.22 in \cite{MELTON1994lagoisconnections}] If $(\ld{L_1}, \rid{\alpha_1}, \ld{\gamma_1}, \rid{M_1})$ and $(\rid{M_1},  \tid{\alpha_2}, \rid{\gamma_2},  \tid{Q})$ are increasing Lagois connections, then the flow defined by the increasing Lagois connection $(\ld{L_1},  \tid{\alpha_2} \circ \rid{\alpha_1} , \ld{\gamma_1} \circ \rid{\gamma_2},  \tid{Q})$ is secure iff 
\begin{align}
    \rid{\gamma_2} \circ \tid{\alpha_2} \circ \rid{\alpha_1}[\ld{L_1}] \subseteq \rid{\alpha_1}[\ld{L_1}]\ and
\label{cond:composelagois1}
\end{align}
\begin{align}
    \rid{\alpha_1} \circ \ld{\gamma_1} \circ \rid{\gamma_2}[\tid{Q}] \subseteq \rid{\gamma_2}[\tid{Q}] 
\label{cond:composelagois2}
\end{align} 
\label{lemma:composinglagoisconnections}
\end{theorem}
\begin{corollary}[Corollary 3.23 in \cite{MELTON1994lagoisconnections}] \label{lemma:composition2}
If $(\ld{L_1}, \rid{\alpha_1}, \ld{\gamma_1}, \rid{M_1})$ and $(\rid{M_1},  \tid{\alpha_2}, \rid{\gamma_2},  \tid{Q})$ are Lagois connections and if either $\rid{\gamma_2}[\tid{Q}] \subseteq \rid{\alpha_1}(\ld{L_1})$ or $\rid{\gamma_2}[\tid{Q}] \supseteq \rid{\alpha_1}(\ld{L_1})$, then $(\ld{L_1}, \tid{\alpha_2} \circ \rid{\alpha_1}, \ld{\gamma_1} \circ \rid{\gamma_2}, \tid{Q})$ is a Lagois connection.
\end{corollary}

We illustrate how an MoU negotiated between \textit{Dorm-Life}  (which has security classes \textit{HouseMaster}, \textit{DiningManager}, \textit{Caretaker}, \textit{Assistant}, and \textit{Student}) and \textit{College} can be composed with the MoU which has been negotiated between the \textit{College} and \textit{University} in Figure \ref{fig:Bi-L-WnD-compose}, to come up with an MoU which allows secure bidirectional information flow between all three domains.

\begin{figure}[!ht]
    \centering
    \begin{tikzpicture}[framed,->,node distance=1cm,on grid]
    \title{W and D}
    \node(T2)  [xshift=3cm] {{\color{red}\scriptsize$\top 2$}};
    \node(T1)    {{\color{blue} \scriptsize$\top 1$}};
    \node(Dir1) [below of = T1] {{\color{blue}\scriptsize$CollegePrincipal$}};
    \node(D1) [below left of = Dir1] {\scriptsize$Dean\ (F)$};
    \node(F1) [below of = D1] {\scriptsize$Faculty$};
    \node(DS1) [below right of = Dir1] {\scriptsize$Dean\ (S)$};
    \node(S1) [below right of = F1] {{\color{blue}\scriptsize$Student$}};
    \node(B1)  [below of=S1]  {{\color{blue}\scriptsize$\bot 1$}};
    \node(Sec2)  [below of = T2] {\scriptsize$Chancellor$};
    \node[align=center] (AS2)  [below of=Sec2] {\scriptsize $Vice$ \\ \scriptsize $Chancellor$};
    \node(Dir2)  [below of=AS2] {{\color{red}\scriptsize$Dean(Colg)$}};
    \node(E2)  [below of=Dir2] {{\color{red}\scriptsize$Univ.Fac.$}};
    \node(S2) [below of = E2] {{\color{red}\scriptsize$Student$}};
    \node(B2)  [below of=S2]  {{\color{red}\scriptsize$\bot 2$}};
    \draw [densely dashed, black,  thick, <->] (F1) to (E2);
    \draw [densely dashed, black,  thick,<->] (B1) to (B2);
    \draw [densely dashed, black,  thick, <->] (T1) to (T2);
    \draw [densely dashed, black,  thick, <->] (S2) to (S1);
    \draw [densely dashed, black,  thick, <->] (Dir1) [bend left = 10] to (Dir2);
    \draw [OliveGreen, thick] (D1) to [bend right = 5] (Dir2);
    \draw [OliveGreen, thick] (DS1) to (Dir2);
    \draw [brown, thick] (AS2) to (T1);
    \draw [brown, thick] (Sec2) to (T1);
    \draw [Fuchsia, ->] (Dir1) to (T1);
    \draw [Fuchsia, ->] (D1) to (Dir1);
    \draw [Fuchsia, ->] (F1) to (D1);
    \draw [Fuchsia, ->] (S1) to (F1);
    \draw [Fuchsia, ->] (S1) to (DS1);
    \draw [Fuchsia, ->] (DS1) to (Dir1);
    \draw [Fuchsia, ->] (B1) to (S1);
    \draw [Fuchsia, ->] (B2) to (S2);
    \draw [Fuchsia, ->] (S2) to (E2);
    \draw [Fuchsia, ->] (E2) to (Dir2);
    \draw [Fuchsia, ->] (Dir2) to (AS2);
    \draw [Fuchsia, ->] (AS2) to (Sec2);
    \draw [Fuchsia, ->] (Sec2) to (T2);
    
    \draw (0,-2.5) [blue] ellipse (1.4cm and 2.7cm);
    \draw (3,-3)[red] ellipse (1.2cm and 3.3cm);
    \draw (-3.5,-2.5) [Emerald] ellipse (1.5cm and 2.7cm);

   \node(T1')  [xshift=-3.5cm]  {{\color{Emerald} \scriptsize$\top 0$}};
   
    \node(Dir1') [below of = T1'] {{\color{Emerald}\scriptsize$HouseMaster$}};
    \node(D1') [below left of = Dir1'] {\scriptsize$Caretaker$};
    \node[align=center](F1') [below of = D1'] {\scriptsize $Assistant$};
    \node[align=center](DS1') [below right of = Dir1'] {\scriptsize $Dining-$ \\ \scriptsize $Manager$};
    \node(S1') [below right of = F1'] {{\color{Emerald}\scriptsize$Student$}};
    \node(B1')  [below of=S1']  {{\color{Emerald}\scriptsize$\bot 0$}};
    \draw [Fuchsia, ->] (Dir1') to (T1');
    \draw [Fuchsia, ->] (D1') to (Dir1');
    \draw [Fuchsia, ->] (F1') to (D1');
    \draw [Fuchsia, ->] (S1') to (F1');
    \draw [Fuchsia, ->] (S1') to (DS1');
    \draw [Fuchsia, ->] (Dir1') to (DS1');
    \draw [Fuchsia, ->] (B1') to (S1');
    
    \draw [densely dashed, black,  thick, <->] (T1') to (T1);
    \draw [densely dashed, black,  thick, <->] (B1') to (B1);
    \draw [densely dashed, black,  thick, <->] (S1') to (S1);
    \draw [densely dashed, black,  thick, <->] (Dir1') to [bend left = 5] (DS1);
    
    \draw [Mahogany,thick, ->] (Dir1) to [bend right = 10] (T1');
    \draw [Mahogany,thick, ->] (F1) to [bend right = 10] (Dir1');
    \draw [Mahogany,thick, ->] (D1) to [bend right = 10](Dir1');
    \draw [Emerald,thick, ->] (D1') to [bend left = 10]  (DS1);
    \draw [Emerald,thick, ->] (F1') to [bend right = 3] (DS1);
    \draw [Emerald,thick, ->] (DS1') to [bend right = 10] (DS1);

    

    \end{tikzpicture}
    \caption{\small Composing Lagois connections.
    Here an MoU negotiated between \textit{Dorm-Life} and \textit{College} is composed with another MoU which has been negotiated between \textit{College} and \textit{University}.
    \label{fig:Bi-L-WnD-compose}} 
\end{figure}
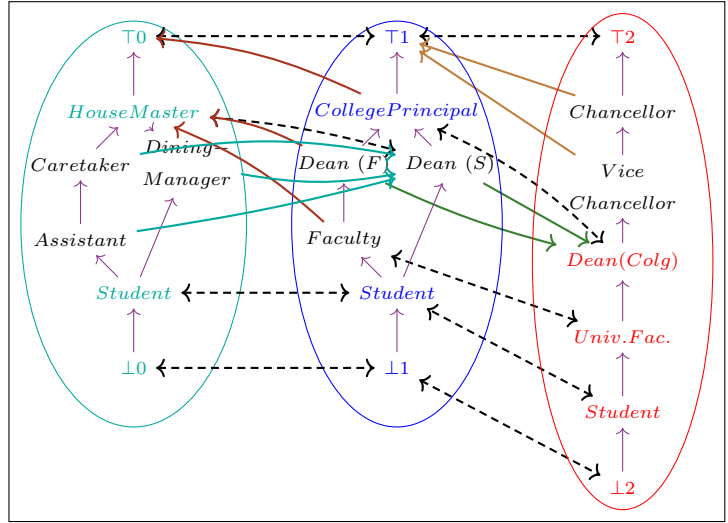
\section{Maintaining MoUs When Security Lattices Change}\label{sec:maintaining-mou}

\subsection{Analysing a Lagois Connection}\label{sec:lagois-decomposition}
Before we discuss how to update an existing MoU, we present Theorem \ref{lemma:decompose}, which provides us a decomposed, analytical view of the anatomy of an increasing Lagois connection. 
For example, the Lagois connection given in Figure \ref{fig:findingLCfromalpha2} 
can be alternatively viewed as shown in Figure \ref{fig:Bi-L-WnD-decomposedView}. 

Let us assume that $(\ld{L}, \rid{\alpha}, \ld{\gamma}, \rid{M})$ is a Lagois connection. Then, let $\ld{L^*} = \ld{\gamma}[\rid{M}]$ and $\rid{M^*} = \rid{\alpha}[\ld{L}]$. Also, let $\ld{r_1} = (\ld{\gamma}  \circ  \rid{\alpha})|_{\ld{L}}^{\ld{L^*}}$, $\rid{r_2} = (\rid{\alpha}  \circ  \ld{\gamma})|^{\rid{M^*}}_{\rid{M}}$, $\rid{i_1} = \rid{\alpha}|^{\rid{M^*}}_{\ld{L^*}}$, $\ld{i_2} = \ld{\gamma}|^{\ld{L^*}}_{\rid{M^*}}$ and let $\ld{e_1}$ be the embedding (inclusion) from $\ld{L^*}$ to $\ld{L}$ and $\rid{e_2}$ be the embedding (inclusion) from $\rid{M^*}$ to $\rid{M}$. Then, $(\ld{L^*}, \rid{i_1}, \ld{i_2}, \rid{M^*})$ is a Lagois isomorphism. This isomorphic substructure is very helpful in reducing the computational effort involved in maintaining the negotiated MoUs based on Lagois connections, as detailed in the sequel.

The key insight is that as long as the inter-domain mappings, i.e., $i_1$ and $i_2$ are not changed, and new security lattices can be connected to old lattices with an increasing Lagois insertion\footnote{A Lagois connection is an insertion if one of the two mappings is injective.}, we will not need to re-negotiate the inter-domain mappings. For example, Figure \ref{fig:Bi-L-WnD-addSC1} shows two such Lagois insertions between old and new security lattices.
\begin{theorem} [Theorem 3.24 in \cite{MELTON1994lagoisconnections}]\label{lemma:decompose}
Every increasing Lagois connection $(\ld{L}, \rid{\alpha}, \ld{\gamma}, \rid{M})$ is a composite\footnote{For easy left-to-right readability, we have used a composition operator $\diamond$ and reversed the order of the components from how they appear in Theorem 3.24 in \cite{MELTON1994lagoisconnections}. } 
$(\ld{L}, \ld{r_1}, \ld{e_1}, \ld{L^*})
 \diamond 
 (\ld{L^*}, \rid{i_1}, \ld{i_2}, \rid{M^*}) 
 \diamond
 (\rid{M^*}, \rid{e_2}, \rid{r_2}, \rid{M})$ where
\begin{enumerate}
    \item $(\ld{L}, \ld{r_1}, \ld{e_1}, \ld{L^*})$ is an increasing Lagois insertion,
    \item $\rid{i_1}$ and $\ld{i_2}$ are isomorphisms that are inverses to each other, and
    \item $(\rid{M^*}, \rid{e_2}, \rid{r_2}, \rid{M})$ is an increasing Lagois insertion.
\end{enumerate}
\end{theorem}

\begin{figure}[!ht]
\centering
    \begin{tikzpicture}[framed,->,node distance=1cm,on grid]
    \title{W and D}
    \node(T2)  [xshift=3cm, yshift=-1cm] {{\color{red}\scriptsize$\top 2$}};
    \node(T1)  [yshift=-1cm]  {{\color{blue} \scriptsize$\top 1$}};
    \node(Dir1) [below of = T1] {{\color{blue}\scriptsize$CollegePrincipal$}};
    \node(F1) [below of = Dir1] {{\color{blue}\scriptsize$Faculty$}};
    \node(B1)  [below of=F1]  {{\color{blue}\scriptsize$\bot 1$}};
    \node(Dir2)  [below of=T2] {{\color{red}\scriptsize$Dean(Colg)$}};
    \node(E2)  [below of=Dir2] {{\color{red}\scriptsize$Univ.Fac.$}};
    \node(B2)  [below of=E2]  {{\color{red}\scriptsize$\bot 2$}};
    \draw [densely dashed, black,  thick, <->] (F1) to (E2);
    \draw [densely dashed, black,  thick,<->] (B1) to (B2);
    \draw [densely dashed, black,  thick, <->] (T1) to (T2);
    \draw [densely dashed, black,  thick, <->] (Dir1) to (Dir2);
    \draw [Fuchsia, ->] (Dir1) to (T1);
    \draw [Fuchsia, ->] (F1) to (Dir1);
    \draw [Fuchsia, ->] (B1) to (F1);
    \draw [Fuchsia, ->] (B2) to (E2);
    \draw [Fuchsia, ->] (E2) to (Dir2);
    \draw [Fuchsia, ->] (Dir2) to (T2);
    \draw (0,-2.5)[blue] ellipse (1.4cm and 2cm);
    \draw (3,-2.5)[red] ellipse (1.2cm and 2cm);

   \node(T2')  [xshift=6cm] {{\color{red}\scriptsize$\top 2$}};
   \node(T1')  [xshift=-3.5cm]  {{\color{blue} \scriptsize$\top 1$}};
   
   \node(Dir1') [below of = T1'] {{\color{blue}\scriptsize$CollegePrincipal$}};
    \node(D1') [below left of = Dir1'] {\scriptsize$Dean\ (F)$};
    \node(F1') [below of = D1'] {{\color{blue}\scriptsize$Faculty$}};
    \node(DS1') [below right of = Dir1'] {\scriptsize$Dean\ (S)$};
    \node(S1') [below right of = F1'] {{\color{blue}\scriptsize$Student$}};
    \node(B1')  [below of=S1']  {{\color{blue}\scriptsize$\bot 1$}};
    \node(Sec2')  [below of = T2'] {\scriptsize$Chancellor$};
    \node[align=center] (AS2')  [below of=Sec2'] {\scriptsize $Vice$ \\ \scriptsize $Chancellor$};
    \node(Dir2')  [below of=AS2'] {{\color{red}\scriptsize$Dean(Colg)$}};
    \node(E2')  [below of=Dir2'] {{\color{red}\scriptsize$Univ.Fac.$}};
    \node(B2')  [below of=E2']  {{\color{red}\scriptsize$\bot 2$}};
   
    \draw [Fuchsia, ->] (Dir1') to (T1');
    \draw [Fuchsia, ->] (D1') to (Dir1');
    \draw [Fuchsia, ->] (F1') to (D1');
    \draw [Fuchsia, ->] (S1') to (F1');
    \draw [Fuchsia, ->] (S1') to (DS1');
    \draw [Fuchsia, ->] (DS1') to (Dir1');
    \draw [Fuchsia, ->] (B1') to (S1');
    \draw [Fuchsia, ->] (B2') to (E2');
    \draw [Fuchsia, ->] (E2') to (Dir2');
    \draw [Fuchsia, ->] (Dir2') to (AS2');
    \draw [Fuchsia, ->] (AS2') to (Sec2');
    \draw [Fuchsia, ->] (Sec2') to (T2');
    
    \draw [densely dashed, black,  thick, <->] (T1') to (T1);
    \draw [densely dashed, black,  thick, <->] (B1') to (B1);
    \draw [densely dashed, black,  thick, <->] (Dir1') to (Dir1);
    \draw [densely dashed, black,  thick, <->] (F1') to [bend right = 5] (F1);
    \draw [OliveGreen,  thick, ->] (S1') to [bend right = 15] (F1);
    \draw [OliveGreen,  thick, ->] (DS1') to (Dir1);
    \draw [OliveGreen,  thick, ->] (D1') to [bend right = 10] (Dir1);
    
    \draw [densely dashed, black,  thick, <->] (T2') to (T2);
    \draw [densely dashed, black,  thick, <->] (B2') to (B2);
    \draw [brown, thick, ->] (Sec2') to (T2);
    \draw [brown, thick, ->] (AS2') to (T2);
    \draw [densely dashed, black,  thick, <->] (Dir2') to (Dir2);
    \draw [densely dashed, black,  thick, <->] (E2') to (E2);
    
    \draw (-3.5,-2.5) [blue] ellipse (1.4cm and 3.1cm);
    \draw (6,-2.5)[red] ellipse (1.3cm and 2.7cm);
    
    \draw [dotted] (6,-1.4) ellipse (0.8cm and 1.2cm);
    \draw [dotted] (-3.5,-1.7) ellipse (1.35cm and 0.3cm);
    \draw [rotate around={45:(-3.8,-3.4)},dotted] (-3.8,-3.4) ellipse (0.4cm and 0.6cm);

    \end{tikzpicture}
    \caption{\small A decomposed view of a Lagois connection.
    Dashed black arrows define permissible flows between budpoints. \label{fig:Bi-L-WnD-decomposedView}} 
\end{figure}
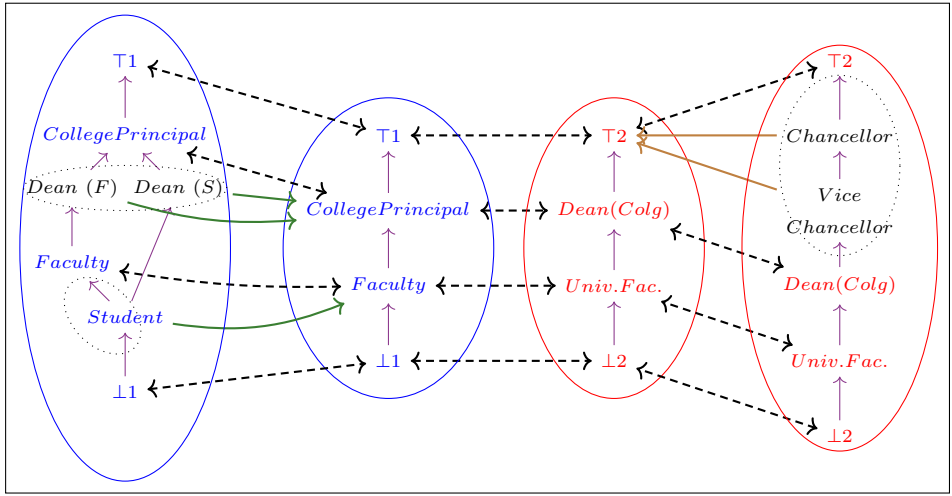

\subsection{Ch-ch-ch-ch-changes}\label{sec:changes}

There are four ways in which security lattices can evolve over time:
\begin{enumerate}
    \item Adding security classes to the existing security lattices;
    \item Removing security classes from the existing security lattices;
    \item Adding new edges to the existing security lattices;
    \item Removing edges from existing security lattices.
\end{enumerate} 

We examine how to re-establish a secure Lagois connection when changes in the lattice structures occur.
The key observation is that we only need to monitor if the original isomorphic substructure mediating the old Lagois connection between the participating security lattices changes or not.

If the participating isomorphic substructure remains unchanged, and one can find an increasing Lagois insertion between the changed  security lattice and the original isomorphic substructure then the MoU need not be re-negotiated. 
The necessary updates of the monotone function can be done independently of the other organisation.

\begin{enumerate}
 \item \textit{No change in the isomorphic sub-structure}:
    \begin{enumerate}
        \item The number of equivalence classes remains the same. 
        This is a simple case. 
        Use Theorem \ref{lemma:decompose} to first find a Lagois insertion between the new security lattice and the isomorphic substructure, and then update the  functions $r_i$ mentioned in Theorem \ref{lemma:decompose}, which can be done independently of the other parts of the diagram. 
        We illustrate this with an example where both security lattices are updated by adding Teaching Assistants (\textit{TAs}) and student \textit{Mentors} in College and department heads (\textit{HOD}) in \textit{University}, as shown in Figure \ref{fig:Bi-L-WnD-addSC1}.
        Using Theorem \ref{lemma:decompose}, we get the updated Lagois connection shown in Figure \ref{fig:Bi-L-WnD-addSC2} --- adding 2 new green arrows on the left, and a new brown arrow on the right.
    \end{enumerate}
    \item \textit{Change in the isomorphic sub-structure}:
    \begin{enumerate}
        \item \textit{Increase in number of equivalence classes}: Check the four conditions of Lagois connections (\textbf{LC1}, \textbf{LC2}, \textbf{LC3} and \textbf{LC4}) for the affected equivalence classes in the security lattices.  
        This is similar to re-negotiating a Lagois connection for those parts of security lattices, as discussed earlier. Use Theorem \ref{lemma:existenceLC} for finding a new Lagois connection.
        \item \textit{Decrease in number of equivalence classes}: Use Corollary \ref{lemma:refine} to check if the new mappings (of which the initial mappings are a refinement) form a Lagois connection (discussed in more detail below). 
    \end{enumerate}
       
\end{enumerate}
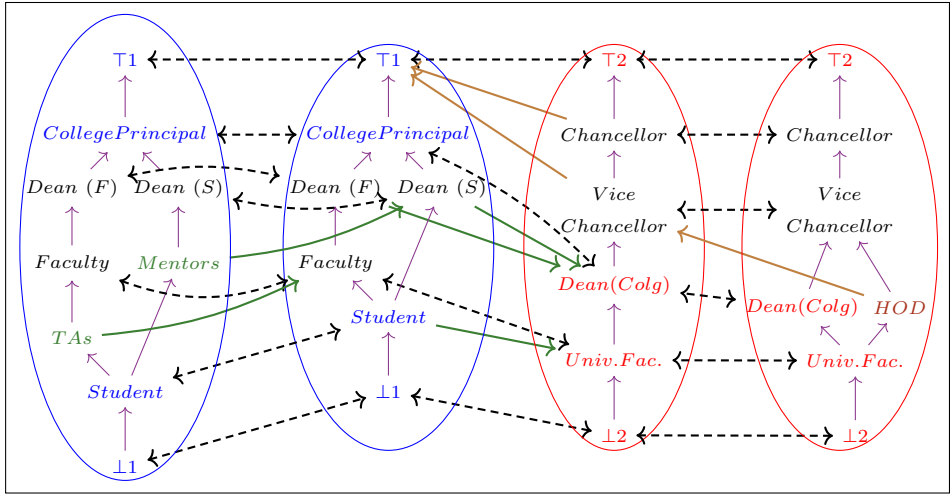
\begin{figure}[!ht]
\centering 
    \begin{tikzpicture}[framed,->,node distance=1cm,on grid]
    \title{W and D}
    \node(T2)  [xshift=3cm] {{\color{red}\scriptsize$\top 2$}};
    \node(T1)    {{\color{blue} \scriptsize$\top 1$}};
    
    \node(Dir1) [below of = T1] {{\color{blue}\scriptsize$CollegePrincipal$}};
    \node(D1) [below left of = Dir1] {\scriptsize$Dean\ (F)$};
    \node(F1) [below of = D1] {\scriptsize$Faculty$};
    \node(DS1) [below right of = Dir1] {\scriptsize$Dean\ (S)$};
    \node(S1) [below right of = F1] {{\color{blue}\scriptsize$Student$}};
    \node(B1)  [below of=S1]  {{\color{blue}\scriptsize$\bot 1$}};
    \node(Sec2)  [below of = T2] {\scriptsize$Chancellor$};
    \node[align=center] (AS2)  [below of=Sec2] {\scriptsize $Vice$ \\ \scriptsize $Chancellor$};
    \node(Dir2)  [below of=AS2] {{\color{red}\scriptsize$Dean(Colg)$}};
    \node(E2)  [below of=Dir2] {{\color{red}\scriptsize$Univ.Fac.$}};
    \node(B2)  [below of=E2]  {{\color{red}\scriptsize$\bot 2$}};
    \draw [densely dashed, black,  thick, <->] (F1) to (E2);
    \draw [OliveGreen, ->, thick] (S1) to (E2);
    \draw [densely dashed, black,  thick,<->] (B1) to (B2);
    \draw [densely dashed, black,  thick, <->] (T1) to (T2);
    \draw [densely dashed, black,  thick, <->] (Dir1) [bend left = 10] to (Dir2);
    \draw [OliveGreen, thick] (D1) to (Dir2);
    \draw [OliveGreen, thick] (DS1) to (Dir2);
    \draw [brown, thick] (AS2) to (T1);
    \draw [brown, thick] (Sec2) to (T1);
    \draw [Fuchsia, ->] (Dir1) to (T1);
    \draw [Fuchsia, ->] (D1) to (Dir1);
    \draw [Fuchsia, ->] (F1) to (D1);
    \draw [Fuchsia, ->] (S1) to (F1);
    \draw [Fuchsia, ->] (S1) to (DS1);
    \draw [Fuchsia, ->] (DS1) to (Dir1);
    \draw [Fuchsia, ->] (B1) to (S1);
    \draw [Fuchsia, ->] (B2) to (E2);
    \draw [Fuchsia, ->] (E2) to (Dir2);
    \draw [Fuchsia, ->] (Dir2) to (AS2);
    \draw [Fuchsia, ->] (AS2) to (Sec2);
    \draw [Fuchsia, ->] (Sec2) to (T2);
    
    \draw (0,-2.5) [blue] ellipse (1.4cm and 2.7cm);
    \draw (3,-2.5)[red] ellipse (1.2cm and 2.7cm);
    

   \node(T2')  [xshift=6cm] {{\color{red}\scriptsize$\top 2$}};
   \node(T1')  [xshift=-3.5cm]  {{\color{blue} \scriptsize$\top 1$}};
   
    \node(Dir1') [below of = T1'] {{\color{blue}\scriptsize$CollegePrincipal$}};
    \node(D1') [below left of = Dir1'] {\scriptsize$Dean\ (F)$};
    \node(F1') [below of = D1'] {\scriptsize$Faculty$};
    \node(DS1') [below right of = Dir1'] {\scriptsize$Dean\ (S)$};
    \node(M1') [below of = DS1'] {{\color{OliveGreen}\scriptsize$Mentors$}};
    \node(TA1') [below of = F1'] {{\color{OliveGreen}\scriptsize$TAs$}};
    \node(S1') [below right of = TA1'] {{\color{blue}\scriptsize$Student$}};
    \node(B1')  [below of=S1']  {{\color{blue}\scriptsize$\bot 1$}};
    \node(Sec2')  [below of = T2'] {\scriptsize$Chancellor$};
    \node[align=center] (AS2')  [below of=Sec2'] {\scriptsize $Vice$ \\ \scriptsize $Chancellor$};
    \node(Dir2')  [xshift=5.5cm,yshift=-3.3cm] {{\color{red}\scriptsize$Dean(Colg)$}};
    \node(HoD2')  [xshift=6.8cm,yshift=-3.3cm] {{\color{Mahogany}\scriptsize$HOD$}};
    \node(E2')  [below right of=Dir2'] {{\color{red}\scriptsize$Univ.Fac.$}};
    \node(B2')  [below of=E2']  {{\color{red}\scriptsize$\bot 2$}};
    \draw [Fuchsia, ->] (Dir1') to (T1');
    \draw [Fuchsia, ->] (D1') to (Dir1');
    \draw [Fuchsia, ->] (F1') to (D1');
    \draw [Fuchsia, ->] (TA1') to (F1');
    \draw [Fuchsia, ->] (S1') to (TA1');
    \draw [Fuchsia, ->] (S1') to (M1');
    \draw [Fuchsia, ->] (M1') to (DS1');
    \draw [Fuchsia, ->] (DS1') to (Dir1');
    \draw [Fuchsia, ->] (B1') to (S1');
    \draw [Fuchsia, ->] (B2') to (E2');
    \draw [Fuchsia, ->] (E2') to (Dir2');
    \draw [Fuchsia, ->] (Dir2') to (AS2');
     \draw [Fuchsia, ->] (E2') to (HoD2');
    \draw [Fuchsia, ->] (HoD2') to (AS2');
    \draw [Fuchsia, ->] (AS2') to (Sec2');
    \draw [Fuchsia, ->] (Sec2') to (T2');
    
    \draw [densely dashed, black,  thick, <->] (T1') to (T1);
    \draw [densely dashed, black,  thick, <->] (B1') to (B1);
    \draw [densely dashed, black,  thick, <->] (Dir1') to (Dir1);
    \draw [densely dashed, black,  thick, <->] (D1') to [bend left = 12] (D1);
    \draw [densely dashed, black,  thick, <->] (F1') to [bend right = 20] (F1);
    \draw [OliveGreen,  thick] (TA1') to [bend right = 10] (F1);
    \draw [OliveGreen,  thick] (M1') to  [bend right = 10](DS1);
    \draw [densely dashed, black,  thick, <->] (S1') to (S1);
    \draw [densely dashed, black,  thick, <->] (DS1') to [bend right = 12] (DS1);
    
    \draw [densely dashed, black,  thick, <->] (T2') to (T2);
    \draw [densely dashed, black,  thick, <->] (B2') to (B2);
    \draw [densely dashed, black,  thick, <->] (Sec2') to (Sec2);
    \draw [densely dashed, black,  thick, <->] (AS2') to (AS2);
    \draw [densely dashed, black,  thick, <->] (Dir2') to (Dir2);
    \draw [densely dashed, black,  thick, <->] (E2') to (E2);
    \draw [brown,  thick] (HoD2') to (AS2);
    
    \draw (-3.5,-2.5) [blue] ellipse (1.4cm and 3.1cm);
    \draw (6,-2.5)[red] ellipse (1.3cm and 2.7cm);
    

    \end{tikzpicture}
    \caption{\small Organisations can add security classes to their lattice structures autonomously as long as they are able to connect the new lattice structures with the old lattice structures (participating in the MoU) via a Lagois insertion.
    Dashed black arrows define permissible flows between budpoints. \label{fig:Bi-L-WnD-addSC1}} 
\end{figure}
\begin{figure}[!ht]
\centering
\begin{tikzpicture}[framed,->,node distance=1cm,on grid]
    \node(T2')  {{\color{blue}\scriptsize$\top 2$}};
   \node(T1')  [xshift=-3.5cm]  {{\color{blue} \scriptsize$\top 1$}};
   
    \node(Dir1') [below of = T1'] {{\color{blue}\scriptsize$CollegePrincipal$}};
    \node(D1') [below left of = Dir1'] {\scriptsize$Dean\ (F)$};
    \node(F1') [below of = D1'] {\scriptsize$Faculty$};
    \node(DS1') [below right of = Dir1'] {\scriptsize$Dean\ (S)$};
    \node(M1') [below of = DS1'] {{\color{OliveGreen}\scriptsize$Mentors$}};
    \node(TA1') [below of = F1'] {{\color{OliveGreen}\scriptsize$TAs$}};
    \node(S1') [below right of = TA1'] {{\color{blue}\scriptsize$Student$}};
    \node(B1')  [below of=S1']  {{\color{blue}\scriptsize$\bot 1$}};
    \node(Sec2')  [below of = T2'] {\scriptsize$Chancellor$};
    \node[align=center] (AS2')  [below of=Sec2'] {\scriptsize $Vice$ \\ \scriptsize $Chancellor$};
    \node(Dir2')  [xshift=-0.5cm,yshift=-3.3cm] {{\color{blue}\scriptsize$Dean(Colg)$}};
    \node(HoD2')  [xshift=0.8cm,yshift=-3.3cm] {{\color{Mahogany}\scriptsize$HOD$}};
    \node(E2')  [below right of=Dir2'] {{\color{blue}\scriptsize$Univ.Fac.$}};
    \node(B2')  [below of=E2']  {{\color{blue}\scriptsize$\bot 2$}};
    \draw [Fuchsia, ->] (Dir1') to (T1');
    \draw [Fuchsia, ->] (D1') to (Dir1');
    \draw [Fuchsia, ->] (F1') to (D1');
    \draw [Fuchsia, ->] (TA1') to (F1');
    \draw [Fuchsia, ->] (S1') to (TA1');
    \draw [Fuchsia, ->] (S1') to (M1');
    \draw [Fuchsia, ->] (M1') to (DS1');
    \draw [Fuchsia, ->] (DS1') to (Dir1');
    \draw [Fuchsia, ->] (B1') to (S1');
    \draw [Fuchsia, ->] (B2') to (E2');
    \draw [Fuchsia, ->] (E2') to (Dir2');
    \draw [Fuchsia, ->] (Dir2') to (AS2');
     \draw [Fuchsia, ->] (E2') to (HoD2');
    \draw [Fuchsia, ->] (HoD2') to (AS2');
    \draw [Fuchsia, ->] (AS2') to (Sec2');
    \draw [Fuchsia, ->] (Sec2') to (T2');
    
    \draw [densely dashed, black,  thick, <->] (F1') to [bend right = 10] (E2');
    \draw [OliveGreen, ->, thick] (TA1') to [bend right = 5] (E2');
    \draw [OliveGreen, ->,thick] (M1') [bend right = 10] to (Dir2');
    \draw [brown, ->,thick] (HoD2') [bend left = 3] to (T1');
    \draw [OliveGreen, ->, thick] (S1') to [bend right = 5] (E2');
    \draw [densely dashed, black,  thick,<->] (B1') to (B2');
    \draw [densely dashed, black,  thick, <->] (T1') to (T2');
    \draw [densely dashed, black,  thick, <->] (Dir1') [bend left = 20] to (Dir2');
    \draw [OliveGreen, thick] (D1') to (Dir2');
    \draw [OliveGreen, thick] (DS1') to (Dir2');
    \draw [brown, thick] (AS2') to (T1');
    \draw [brown, thick] (Sec2') to (T1');
    
    \draw (-3.5,-2.5) [blue] ellipse (1.4cm and 3.1cm);
    \draw (0,-2.5)[red] ellipse (1.3cm and 2.7cm);
    
    \draw [rotate around={25:(0,-1.4)}, dotted] (0,-2) ellipse (0.8cm and 1.8cm);
    \draw[rotate around={70:(-3.4,-2)}, dotted] (-3.4,-2) ellipse (0.7cm and 1.4cm);
    \draw[rotate around={40:(-4.2,-3.4)}, dotted] (-4.2,-4) ellipse (0.55cm and 1.15cm);
    
    \end{tikzpicture}
    \caption{\small A \emph{new} viable increasing Lagois connection created by the composition of old security lattices with new security lattices. 
    Dashed black arrows define permissible flows between budpoints. \label{fig:Bi-L-WnD-addSC2}} 
\end{figure}
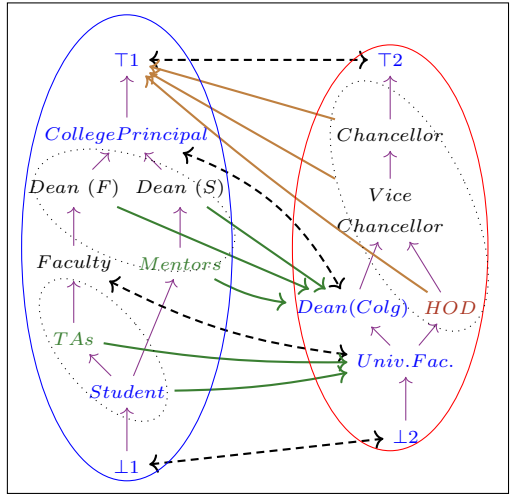
\subsubsection{Coarsening of inter-domain mappings.}
When the isomorphic structure is changed, the MoU has to be renegotiated.  When deleting a security class, we can always ``up-classify'' information pertaining to that class, to cause a stricter flow of information to a higher security class in the other lattice. 
When coarsening a lattice, we would like to re-establish a Lagois connection that respects the previous Lagois connection as much as possible (maintaining legacy) while making the necessary up-classifications where necessary.


Let $(\ld{L},\rid{\alpha}, \ld{\gamma}, \rid{M})$ be an existing secure Lagois connection.
Suppose we need to replace $\rid{\alpha}$ by another order-preserving function $\rid{\alpha'}: \ld{L} \rightarrow \rid{M}$. 
We look for a suitable 
map $\ld{\gamma'}: \rid{M} \rightarrow \ld{L}$ such that $(\ld{L},\rid{\alpha'}, \ld{\gamma'}, \rid{M})$ is a Lagois connection using the following result:
\begin{proposition}[Proposition 2.4 in \cite{melton1991connections}] \label{def:connection}
If $\rid{\alpha_1}: \ld{L} \rightarrow \rid{M}$ is an order-preserving map with semi-inverse $\ld{\gamma_1}: \rid{M} \rightarrow \ld{L}$, i.e.,
$\rid{\alpha_1} \circ \ld{\gamma_1} \circ \rid{\alpha_1} = \rid{\alpha_1}$, then $(\ld{L},\rid{\alpha_1}, \ld{\gamma_1} \circ \rid{\alpha_1} \circ \ld{\gamma_1}, \rid{M})$ is a (Lagois) connection.
\end{proposition}

As mentioned above, when the isomorphic structure is disturbed due to coarsening a lattice, we would like to retain as much of original functions as possible.
We can re-establish a Lagois connection making minimal changes to the old functions $\rid{\alpha}$ and $\ld{\gamma}$ Lagois  by considering an (order-preserving)
$\rid{\alpha'}$ which satisfies the assumptions:
\begin{enumerate}
    \item \label{maintain-cond2} $Ker(\rid{\alpha}) \subseteq Ker(\rid{\alpha'})$\footnote{$Ker(\rid{\alpha})$ is the equivalence relation on $\ld{L}$ given by $(\ld{a},\ld{b}) \in Ker(\rid{\alpha})$ iff $\rid{\alpha}(\ld{a}) = \rid{\alpha}(\ld{b})$. For $\ld{b} \in \ld{L}$, we denote the equivalence class containing $\ld{b}$ by $[\ld{b}]_{\rid{\alpha}}$, as the equivalence relation on $\ld{L}$ is defined by $\rid{\alpha}$.}
    \item \label{maintain-cond3} for each $\ld{l} \in \ld{L}, \rid{\alpha'}(\ld{l}) = \rid{\alpha}(\ld{l^*})$ when $\ld{l^*}$ is the largest element in $\{\ld{l_1} \in \ld{L}| \rid{\alpha'}(\ld{l_1}) = \rid{\alpha'}(\ld{l})\}.$
\end{enumerate}

Typically we would consider as a candidate an $\rid{\alpha'}$ that induces as fine a coarsening of the kernel of $\rid{\alpha}$ as possible while ensuring the second condition (\textit{i.e.}, mapping elements of equivalence classes induced by $\rid{\alpha'}$ (which should be closed with respect to joins) according to how $\rid{\alpha}$ mapped the largest element in that class). 
\begin{corollary}[Corollary 3.16 in \cite{MELTON1994lagoisconnections}] \label{lemma:refine}
Let $(\ld{L},\rid{\alpha}, \ld{\gamma}, \rid{M})$ be a Lagois connection, and let $\rid{\alpha'}\colon \ld{L} \rightarrow \rid{M}$ be an order-preserving map such that
\begin{enumerate}
    \item $\equiv_{\rid{\alpha}}$ is a refinement of $\equiv_{\rid{\alpha'}}$
    (i.e., $Ker(\rid{\alpha}) \subseteq Ker(\rid{\alpha'})$)
    and
    \item for each $\ld{l} \in \ld{L}$, the equivalence class $[\ld{l}]_{\rid{\alpha'}}$ has a largest element -- call it $\ld{l^*}$ -- with $\rid{\alpha'}(\ld{l}) = \rid{\alpha}(\ld{l^*})$.
\end{enumerate}
Then $(\ld{L}, \rid{\alpha'}, \ld{\gamma}  \circ \rid{\alpha'} \circ \ld{\gamma}, \rid{M})$ is a Lagois connection.
\end{corollary}

\section{Securely Connecting  Decentralised Label Models}\label{sec:connecting-DLM}


In this section, we show how Lagois Connections can be used to ensure secure information flow between two organisations, both of which have employed the decentralized label model  (DLM) of Myers \cite{Myers1997-ss} for SIF. 
This illustrates how the Lagois framework can extend autonomy in IFC within an organisation to secure cross-organisational decentralised control.
We show that if a Lagois Connection $LC$ has been established between the Principals hierarchies\footnote{We prefer the term ``principals hierarchy" to ``principal hierarchy'' for grammatical reasons.} of both organisations, we can establish a Lagois connection $\widehat{LC}$ between the security lattices formed by the  labels derived from the respective Principals hierarchies.
The second result we show is that the declassification rule \cite{Myers1997-ss} does not introduce any insecure flows, even when exchanging information between domains.  

\subsection{The Decentralised Label Model}\label{sec:DLM-summary}

\paragraph{Principals Hierarchy} \ \ 
DLM based systems \cite{myers1999jflow, liu2009fabric, liu2017fabric} use abstract \textit{principals} to represent entities that can trust or be trusted,\textit{ e.g.}, users, roles, groups, organizations, privileges, etc. 
Principals express trust via \textit{acts-for} relations \cite{Myers1997-ss}. 
If a principal $p$ acts-for a principal $q$, then $q$ trusts $p$ completely and $p$ may perform any action (read/
 declassify) that $q$ may perform (written $p \succeq q$). The acts-for relation is a \textit{pre-order} and the \textit{principals hierarchy} refers to the set of principals under the acts-for ordering.
The operators $\conjp$ 
and $\disjp$ 
can be used to form conjunction and disjunction of principals in more elaborate principals hierarchies.
The conjunctive principal $p \conjp q$ represents the joint authority of $p$ and $q$, and acts for both: $p \conjp q \succeq p$ and $p \conjp q \succeq q$. 
The disjunctive principal $p \disjp q$ represents the disjoint authority of $p$ and $q$, and is acted for by both: $p \succeq p \disjp q$ and $q \succeq p \disjp q$. 
For convenience we include the most
restrictive and least restrictive principals, denoted as $\top$ and $\bot$ respectively. 

\paragraph{Label Model} \ \ 
The security policies in DLM are expressed using \textit{labels}: each label is the \textit{conjunction} of a \emph{set of policies} each of which expresses privacy\footnote{As observed by Denning, the analysis for the \textit{integrity} of \textit{written} values is dual to the analysis for privacy of values read.} requirements in terms of principals \cite{myers-phd-tr-award}.
In the DLM framework, a privacy policy has two parts: an \textit{owner}, and a \textit{set of readers}, and is written in the form ``owner: readers''. 
The owner of a policy is a principal whose data has contributed to constructing the value that is labeled by this policy. 
The readers of a policy are a set of principals who are permitted by the owner to read the value so labelled.

For example, in the
label $L = \{o_1 : r_3, r_4;\ o_2 : r_4, r_5\}$, there are two policies (semicolons are separators) -- 
one owned by owner $o_1$, which permits the set of readers $\{r_3, r_4\}$, and
the second owned by owner $o_2$ that permits the set of readers $\{r_4, r_5\}$. 
A principal wishing to access an object is required to satisfy \textit{all} confidentiality policy components in the object's label to be able to learn that object's value.
Thus each policy can be viewed as a constraint placed by the owner on what flows are permitted between principals, and a principal who is not the owner of any policy labelling a value places no constraints on allowed flows for that value. 

If a policy $K$ is part of the label $L$ (i.e., $K \in L$), then $\textbf{o}(K): \mathit{policy} \rightarrow \mathit{principal}$ denotes the owner of that policy, and $\textbf{r}(K): \mathit{policy} \rightarrow (\mathit{principal}~\textit{set})$ denotes the set of readers specified by that policy. 
The functions $\textbf{o}$ and $\textbf{r}$ completely characterize a label. 

These policies then can be \textit{modified} safely by the individual owners -- a form of safe \textit{decentralised} declassification. 
An owner may add readers to the reader set of its policy in a label, or remove the entire policy, effectively allowing all readers. 
Arbitrary declassification is not possible because flow policies of other principals remain in force.

\paragraph{Derived Information Flow Lattice} \ \ 
A pre-ordering relation on labels is derived from the acts-for relation $\succeq$ \cite{myers-phd-tr-award}.
 
We write $P \vdash L_1 \sqsubseteq L_2$, when $L_1$ is less or equal to $L_2$, given a principals hierarchy $P$.
The relation
$P \vdash L_1 \sqsubseteq L_2$ is defined formally in \cite{myers-phd-tr-award}, as shown in Figure \ref{fig:completerelabelingrule}.
\begin{figure}[!ht]
    \centering
\begin{equation}
    \boxed{
        \begin{array}{rcl}
        P \vdash L_1 \sqsubseteq L_2 & \Longleftrightarrow & \forall(I \in L_1) \exists(J \in L_2) P \vdash I \sqsubseteq J\\
        P \vdash I \sqsubseteq J
        & \Longleftrightarrow & P \vdash \textbf{o}(J) \succeq \textbf{o}(I) ~\wedge \\
        && ~~\forall(r_j \in \textbf{r}(J)) [ P \vdash r_j \succeq \textbf{o}(I) \vee \exists(r_i \in \textbf{r}(I)) P \vdash r_j \succeq r_i ]\\
        & \Longleftrightarrow & P \vdash \textbf{o}(J) \succeq \textbf{o}(I)~ \wedge \\
        && ~~\forall(r_j \in \textbf{r$^+$}(J))~ \exists(r_i \in \textbf{r$^+$}(I)) ~P \vdash r_j \succeq r_i\\
        \end{array}
    }
\end{equation}
\caption{Definition of complete relabeling rule ($\sqsubseteq$)}
    \label{fig:completerelabelingrule}
\end{figure}
The definition says that given the principals hierarchy $P$, it is safe to relabel information tagged with label $L_1$ to $L2$, if each policy 
$I \in L_1$ is subsumed by a policy $J \in L_2$. Policy $J$ subsumes $I$ when the owner in $J$ ``acts-for'' the owner of $I$, and for every reader $r_j$ (effectively) permitted by policy $J$, there is a reader $r_i$ (effectively) permitted by $I$ such that $r_j$ can act for $r_i$.
Note that each policy component  of a label (e.g., $I,J$) can also be considered a label, so writing $P \vdash I \sqsubseteq J$ is only mild abuse of notation. 
The relation $\sqsubseteq$ is a pre-order, i.e, reflexive and transitive, but not necessarily anti-symmetric. 
However, since the acts-for pre-order $\succeq$ supports join and meet operations, and because the semantics of labels is given in terms of \textit{sets} of permitted flows, we can construct an \textit{information flow lattice} (IFL) by 
 defining an equivalence relation $L_1 \equiv_{\sqsubseteq} L_2 ~\Longleftrightarrow~ (L_1 \sqsubseteq L_2 \text{ and } L_2 \sqsubseteq L_1)$ and taking equivalence classes to be lattice elements. 
 The join and meet on this information flow lattice are called \textit{label join} and \textit{label meet} (written as $\sqcup$ and $\sqcap$ respectively). 
 Since labels are sets of policies that must be together satisfied, we get the so-called ``\textit{Join rule}'', namely $L_1 \sqcup L_2 ~=~ L_1 \cup L_2$.
 The least label and greatest label are $\{\}$ and $\{\top\colon\}$ respectively. 

Myers presents a rule for safe \textit{relabelling by declassification} \cite{myers-phd-tr-award}:
Let $A$ be a set of principals in the current authority.
Let $L_A = \bigcup_{p \in A} \{p \colon \}$. 
Then $L_1$ can be safely declassified to $L_2$ if
$L_1 \sqsubseteq L_2 \sqcup L_A$.

\subsection{Lagois Connections on Principals Hierarchies
and Derived IFLs}\label{sec:lagois-DLM-IFL}
Assume that two organizations with their own principals hierarchies
$\ld{P_L}$ and $\rid{P_R}$ negotiate an increasing Lagois Connection 
\(
    LC = (\ld{P_L}, \rid{\alpha}, \ld{\gamma}, \rid{P_R})
    \)
between their principals hierarchies, $\rid{\alpha}: \ld{P_L} \rightarrow \rid{P_R}$ and $\ld{\gamma}: \rid{P_R} \rightarrow \ld{P_L}$.
This Lagois connection may be considered as a coupling between the input and outputs channels of the respective individual domains, together with certain strong information flow guarantees.

We show that there exists an increasing Lagois connection, $\widehat{LC}$, between the information flow lattices of the two organisations which is derived from the Lagois Connection between their principals hierarchies (see Figure \ref{fig:LCcommutesDLM}). 
\begin{definition}[LC on IFLs]\label{Def:LC_IFL}
\ \ Define 
$
    \widehat{LC} = ( \ld{\mathit{IFL}_L}, \rid{\hat{\alpha}}, \ld{\hat{\gamma}}, \rid{\mathit{IFL}_R} )
$, 
where $\ld{\mathit{IFL}_L}$ and $\rid{\mathit{IFL}_R}$ are the corresponding \textit{information flow lattices}, and
$\rid{\hat{\alpha}}: \ld{\mathit{IFL}_L} \rightarrow \rid{\mathit{IFL}_R}$
and 
$\ld{\hat{\gamma}}:  \rid{\mathit{IFL}_R} \rightarrow
\ld{\mathit{IFL}_L}$ are specified as: 
\begin{flalign*}
&\rid{\hat{\alpha}}(\ld{L_l}) = \bigcup_{\ld{I_l} \in \ld{L_l}}\ 
\rid{\hat{\alpha}}(\ld{I_l}) \\
&\rid{\hat{\alpha}}(\ld{I_l}) =\ \{ \langle \rid{\alpha}(\textbf{o}(\ld{I_l})): \{\rid{\alpha}(\ld{r_l})\ |\ \ld{r_l} \in \textbf{r}(\ld{I_l})\} \rangle \} \\
&\ld{\hat{\gamma}}(\rid{L_r}) = \bigcup_{\rid{I_r} \in \rid{L_r}}\ 
\ld{\hat{\gamma}}(\rid{I_r}) \\
&\ld{\hat{\gamma}}(\rid{I_r}) =\ \{ \langle \ld{\gamma}(\textbf{o}(\rid{I_r})): \{\ld{\gamma}(\rid{r_r})\ |\ \rid{r_r} \in \textbf{r}(\rid{I_r})\} \rangle \}
\end{flalign*}
\end{definition}

Observe that $\rid{\hat{\alpha}}$ and $\ld{\hat{\gamma}}$ distribute homomorphically over joins:
\begin{align}
 \rid{\hat{\alpha}}(\ld{L_1} \ld{\sqcup} \ld{L_2}) ~=~
 \rid{\hat{\alpha}}(\ld{L_1} \cup \ld{L_2}) ~=~
 \rid{\hat{\alpha}}(\ld{L_1}) ~\rd{\sqcup}~
  \rid{\hat{\alpha}}(\ld{L_2})
\label{LagoisLabelDistr1} 
\end{align}
\begin{align}
  \ld{\hat{\gamma}}(\rd{L_1} \rd{\sqcup} \rd{L_2}) ~=~
 \ld{\hat{\gamma}}(\rd{L_1} \cup \rd{L_2}) ~=~
 \ld{\hat{\gamma}}(\rd{L_1}) ~\ld{\sqcup}~
  \ld{\hat{\gamma}}(\rd{L_2})
  \label{LagoisLabelDistr2} 
\end{align}

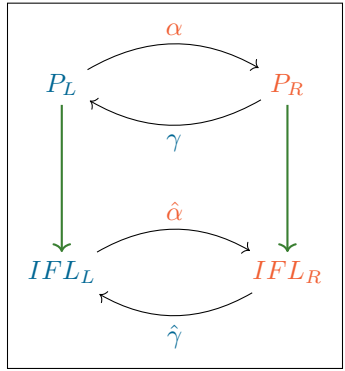
\begin{figure}[!ht]
    \centering
    \begin{tikzpicture}[framed,->,node distance=1cm,on grid]
    \node(T2)   {$\rid{P_R}$};
    \node(T1)  [xshift=-3cm]  {$\ld{P_L}$};
    \node(B1)  [xshift=-3cm, yshift=-2.5cm]  {$\ld{IFL_L}$};
    \node(B2)  [yshift=-2.5cm]  {$\rid{IFL_R}$};
    
    \draw[every loop] 
    (T1) edge[bend left, auto=left] node {$\rid{\alpha}$} (T2) 
    (T2) edge[bend left, auto=left] node {$\ld{\gamma}$} (T1) ;
    
    \draw[every loop] 
    (B1) edge[bend left, auto=left] node {$\rid{\hat{\alpha}}$} (B2) 
    (B2) edge[bend left, auto=left] node {$\ld{\hat{\gamma}}$} (B1) ;
    \draw [OliveGreen, thick, =>] (T1) to (B1);
    \draw [OliveGreen, thick, =>] (T2) to (B2);
    \end{tikzpicture}
    \caption{A Lagois connection between two principals hierarchies induces a Lagois connection between the corresponding Information Flow Lattices}
    \label{fig:LCcommutesDLM}
\end{figure}

\begin{theorem}\label{Thm:LC-PH-LC-IFL}
Let $LC = (\ld{P_L}, \rid{\alpha}, \ld{\gamma}, \rid{P_R})$ be an increasing Lagois connection and $\ld{\mathit{IFL}_L}$ be derived from $\ld{P_L}$ and $\rid{\mathit{IFL}_R}$ be derived from $\rid{P_R}$, as described earlier\footnote{Note $\rid{\alpha}: \ld{P_L} \rightarrow \rid{P_R}$ and $\ld{\gamma}: \rid{P_R} \rightarrow \ld{P_L}$ are total functions by 
Theorem \ref{prop:fnguniquelydetermin} and Corollary \ref{cor:existenceLC2}.}. 
Then  $\widehat{LC} = (\ld{\mathit{IFL}_L}, \rid{\hat{\alpha}}, \ld{\hat{\gamma}}, \rid{\mathit{IFL}_R})$ is also an increasing Lagois connection.
\end{theorem}
\begin{proof}
We prove the following properties for $\rid{\hat{\alpha}}$ and $\ld{\hat{\gamma}}$ using Definition \ref{Def:LC_IFL}:
\begin{enumerate}
    \item Monotonicity
    \begin{enumerate}
    \item \label{IFL:prop1}If $\ld{P_L} \vdash \ld{L_{l1}} \ld{\sqsubseteq} \ld{L_{l2}}$ then $\rid{P_R} \vdash \rid{\hat{\alpha}}(\ld{L_{l1}})\ \rd{\sqsubseteq}\ \rid{\hat{\alpha}}(\ld{L_{l2}})$
    \item \label{IFL:prop2}If $\rid{P_R} \vdash \rid{L_{r1}} \rd{\sqsubseteq} \rid{L_{r2}}$ then $\ld{P_L} \vdash \ld{\hat{\gamma}}(\rid{L_{r1}})\ \ld{\sqsubseteq}\ \ld{\hat{\gamma}}(\rid{L_{r2}})$
    \end{enumerate}
    \item Increasing 
    \begin{enumerate}
    \item \label{IFL:prop3} $\ld{P_L} \vdash \ld{L_{l1} \sqsubseteq \hat{\gamma}}(\rid{\hat{\alpha}}(\ld{L_{l1}})),~~~ \forall \ld{L_{l1}} \in \ld{\mathit{IFL}_L}$
    \item \label{IFL:prop4} $\rid{P_R} \vdash \rid{L_{r1}} \rd{\sqsubseteq}\ \rid{\hat{\alpha}}(\ld{\hat{\gamma}}(\rid{L_{r1}})), \forall \rid{L_{r1}} \in \rid{\mathit{IFL}_R}$
    \end{enumerate}
    \item Identity/ Equality/ Fixed Points
    \begin{enumerate}
    \item \label{IFL:prop5} $\rid{P_R} \vdash \rid{\hat{\alpha}}(\ld{L_{l1}}) \equiv_{\rd{\sqsubseteq}} \rid{\hat{\alpha}}(\ld{\hat{\gamma}}(\rid{\hat{\alpha}}(\ld{L_{l1}}))), ~~~\forall \ld{L_{l1}} \in \ld{\mathit{IFL}_L}$
    \item \label{IFL:prop6} $\ld{P_L} \vdash \ld{\hat{\gamma}}(\rid{L_{r1}}) \equiv_{\ld{\sqsubseteq}} \ld{\hat{\gamma}}(\rid{\hat{\alpha}}(\ld{\hat{\gamma}}(\rid{L_{r1}}))), \forall \rid{L_{r1}} \in \rid{\mathit{IFL}_R}$
    \end{enumerate}
\end{enumerate}
    We show proofs for properties \ref{IFL:prop1}, \ref{IFL:prop3} and \ref{IFL:prop5}, as others are similar. We use notation $\textbf{r}^+(I) = \textbf{r}(I) \cup\ \textbf{o}(I)$ \cite{myers-phd-tr-award}. 
\begin{enumerate}
    \item[(\ref{IFL:prop1})]
    Given $LC =\  (\ld{P_L}, \rid{\alpha}, \ld{\gamma}, \rid{P_R})$ is an increasing Lagois connection, we know that $\rid{\alpha}$ is an order-preserving function. 
    So, if $\ld{p_i} \preceq \ld{p_j}$ then $\rid{\alpha}(\ld{p_i})\ \rd{\preceq}\ \rid{\alpha}(\ld{p_j})$, for all principals $\ld{p_i,p_j} \in \ld{P_L}$. 
    
    \begin{alignat}{3}
        & \ld{P_L} \vdash \ld{L_{l1} \sqsubseteq L_{l2}} \Longleftrightarrow 
        \forall(\ld{I} \in \ld{L_{l1}}) \exists(\ld{J} \in \ld{L_{l2}}) \ld{P_L} \vdash \ld{I \sqsubseteq J} &&~(\text{Def of} \sqsubseteq) \label{given1a} \\
        & \ld{P_L} \vdash \ld{I \sqsubseteq J} \equiv \ld{P_L} \vdash  \textbf{o}(\ld{I}) \ld{\preceq} \textbf{o}(\ld{J})  ~\wedge~ \\ &\hspace{1.5cm}\forall(\ld{r_j}\in \textbf{r}^+(\ld{J}))\ \exists(\ld{r_i} \in \textbf{r}^+(\ld{I}))\ \ld{P_L} \vdash \ld{r_i \preceq r_j} &&~(\text{Def of} \sqsubseteq) \label{1agivenrd}
    \end{alignat}
As $\rid{\alpha}$ is order-preserving, we get  \label{1alphamonotone}
    \begin{alignat}{3}    
        & \rid{P_R} \vdash \rid{\alpha}(\textbf{o}(\ld{I})) \rd{\preceq} \rid{\alpha}(\textbf{o}(\ld{J})) ~\wedge \\ 
        &\hspace{1cm}\forall(\ld{r_j}\in \textbf{r}^+(\ld{J}))\ \exists(\ld{r_i} \in \textbf{r}^+(\ld{I}))\ \rid{P_R} \vdash \rid{\alpha}(\ld{r_i}) \rd{\preceq} \rid{\alpha}(\ld{r_j}) &&~~(\ref{1agivenrd},\ref{1alphamonotone}) \label{1aa}\\ 
        &\rid{P_R} \vdash \textbf{o}(\rid{\hat{\alpha}}(\ld{I})) \rd{\preceq} \textbf{o}(\rid{\hat{\alpha}}(\ld{J})) \wedge \\ 
        &\hspace{1cm}\forall(\rd{r_j}\in \textbf{r}^+(\rid{\hat{\alpha}}(\ld{J})))\ \exists(\rd{r_i} \in \textbf{r}^+(\rid{\hat{\alpha}}(\ld{I})))~ \rid{P_R} \vdash \rd{r_i} \rd{\preceq} \rd{r_j} &&~(\ref{1aa},\text{Def}~ \ref{Def:LC_IFL}) \label{1alphanew} \\
        & \implies \rid{P_R} \vdash \rid{\hat{\alpha}}(\ld{I}) \rd{\sqsubseteq} \rid{\hat{\alpha}}(\ld{J}) \label{1alph} &&~(\ref{1alphanew},\text{Def}~ \ref{Def:LC_IFL})\\
        & \forall(\ld{I}\in \ld{L_{l1}})\ \exists(\ld{J} \in \ld{L_{l2}})\ \rid{P_R} \vdash \rid{\hat{\alpha}}(\ld{I})\ \rd{\sqsubseteq} \rid{\hat{\alpha}}(\ld{J}) && ~(\ref{given1a},\ref{1alph}) \label{1anewpolicies} && \\
        &\rid{P_R} \vdash \rid{\hat{\alpha}}(\ld{L_{l1}})\ \rd{\sqsubseteq} \rid{\hat{\alpha}}(\ld{L_{l2}}) && ~(\ref{1anewpolicies},\text{Def}~ \ref{Def:LC_IFL})
    \end{alignat}
    Hence we have proved Monotonicity.
    \item[(\ref{IFL:prop3})]
 Assuming $\ld{\gamma} \circ \rid{\alpha}$ is increasing we have for any $\ld{p} \in \ld{P_L}$:
    \begin{alignat}{3}
    & \ld{P_L} \vdash \ld{p} ~\ld{\preceq}~ \ld{\gamma}(\rid{\alpha}(\ld{p})), ~~
    (\forall \ld{p} \in \ld{L}) && && \\
    & \forall(\ld{I}\in \ld{L_{l1}})\ \ld{P_L} \vdash \textbf{o}(\ld{I}) \ld{\preceq} \ld{\gamma}(\rid{\alpha}(\textbf{o}(\ld{I}))   ~\wedge && 
    \\
    & \hspace{1cm}\forall(\ld{r_i}\in \textbf{r}^+(\ld{I})\ \ld{P_L} \vdash \ld{r_i} ~\ld{\preceq}~ \ld{\gamma}(\rid{\alpha}(\ld{r_i}))   && 
    \\
    & \forall(\ld{I}\in \ld{L_{l1}})\ \ld{P_L} \vdash 
    \textbf{o}(\ld{I}) ~\ld{\preceq}~
    \textbf{o}(\ld{\hat{\gamma}}(\rid{\hat{\alpha}}(\ld{I})))~  \wedge && ~(\text{Def}~\ref{Def:LC_IFL})\\
    & \hspace{1cm}\forall(\ld{r_j}\in \textbf{r}^+(\hat{\gamma}(\rid{\hat{\alpha}}(\ld{I})))\ \exists(\ld{r_i} \in \textbf{r}^+(\ld{I})) \ \ld{P_L} \vdash \ld{r_i} \ld{\preceq} \ld{r_j} && ~(\text{Def}~ \ref{Def:LC_IFL})\\
    & \forall(\ld{I}\in \ld{L_{l1}})\ \ld{P_L} \vdash \ld{I} ~\ld{\sqsubseteq}~ \ld{\hat{\gamma}}(\rid{\hat{\alpha}}(\ld{I})) && ~(\text{Def of }\sqsubseteq)\\
    & \forall(\ld{I}\in \ld{L_{l1}}) \exists(\ld{J} \in \ld{\hat{\gamma}}(\rid{\hat{\alpha}}(\ld{L_{l1}}))\ \ld{P_L} \vdash \ld{I} ~\ld{\sqsubseteq}~ \ld{J} && ~(\text{Def of} \sqsubseteq)\\
    & \ld{P_L} \vdash \ld{L_{l1}} ~\ld{\sqsubseteq}~ \ld{\hat{\gamma}}(\rid{\hat{\alpha}}(\ld{L_{l1}})) && ~(\text{Def of} \sqsubseteq)
    \end{alignat}
    Hence we have proved $\hat{\ld{\gamma}} \circ \rid{\hat{\alpha}}$ is increasing.

    \item[(\ref{IFL:prop5})] We have by Definition \ref{Def:LC_IFL}, $\forall \ld{I_l} \in \ld{L_l}$
    \begin{align*}
        &\rid{\hat{\alpha}}(\ld{I_l}) =\ \langle \rid{\alpha}(\textbf{o}(\ld{I_l})): \{\rid{\alpha}(\ld{r})\ |\ \ld{r} \in \textbf{r}(\ld{I_l})\} \rangle,  \\ \\
        &\rid{\hat{\alpha}}(\ld{\hat{\gamma}}(\rid{\hat{\alpha}}(\ld{I_l}))) = \langle \rid{\alpha}(\ld{\gamma}(\rid{\alpha}(\textbf{o}(\ld{I_l})))): \{\rid{\alpha}(\ld{\gamma}(\rid{\alpha}(\ld{r}))) ~|~ \ld{r} \in \textbf{r}(\ld{I_l})\} \rangle
    \end{align*}
    As $LC = (\ld{P_L}, \rid{\alpha}, \ld{\gamma}, \rid{P_R})$ is an increasing Lagois connection, we know that $\rid{\alpha} \circ \ld{\gamma} \circ \rid{\alpha} = \rid{\alpha}$. 
    Therefore, 
$
    \rid{\hat{\alpha}}(\ld{L_{l1}}) \equiv_{\rd{\sqsubseteq}} \rid{\hat{\alpha}}(\ld{\hat{\gamma}}(\rid{\hat{\alpha}}(\ld{L_{l1}}))),\ \forall \ld{L_{l1}} \in \ld{\mathit{IFL}_L}
$. \\
    Hence proved.
\end{enumerate}
\end{proof}

Assume the principals hierarchies in two domains are static and Laogis-connected.
Let $\ld{A}$ and $\rd{A}$ represent two Lagois-connected sets of principals under which authority two communicating processes are operating in their respective domains. 
Then the safe declassification rule continues to apply even with bidirectional communication between the two domains.
\begin{corollary}\label{corollary:safe-declassification}
Let $LC = (\ld{P_L}, \rid{\alpha}, \ld{\gamma}, \rid{P_R})$ be an increasing Lagois connection and 
let $\widehat{LC} = (\ld{\mathit{IFL}_L}, \rid{\hat{\alpha}}, \ld{\hat{\gamma}}, \rid{\mathit{IFL}_R})$ be the derived Lagois connection between the corresponding IFLs. 
Let $\ld{A}$ and $\rd{A}$ be sets of principals such that $\rid{\hat{\alpha}}[\ld{A}] = \rd{A}$ and
$\ld{\hat{\gamma}}[\rd{A}] = \ld{A}$.
Then (a) $\ld{L_1}$ can be safely declassified to
$\ld{\hat{\gamma}}(\rd{L_2})$ iff
$\rid{\hat{\alpha}}(\ld{L_1})$ can be safely declassified to $\rd{L_2}$, and (b)
$\rd{L_1}$ can be safely declassified to
$\rid{\hat{\alpha}}(\ld{L_2})$ iff
$\ld{\hat{\gamma}}(\rd{L_1})$ can be safely declassified to $\ld{L_2}$.
\end{corollary}
\begin{proof}
Straightforward from monotonicity, expansiveness, homomorphic distribution over joins, and the assumed Lagois connection between $\ld{A}$ and $\rd{A}$.
\end{proof}

\section{Related Work}\label{sec:related}
The only cited use of the notion of Lagois connections \cite{MELTON1994lagoisconnections} in computer science of which we are aware is the work of Huth \cite{huth1993equivalence} in establishing the correctness of programming language implementations.
To our knowledge, our work is the first to propose their use in secure information flow control.

Abstract Interpretation and type systems \cite{cousot1997types-as-ai} have been used in secure flow analyses, \textit{e.g.},  \cite{cortesi2015datacentricsemantics, cortesi2018} and  \cite{zanotti2002sectypingsbyai}, where security types are defined using Galois connections employing, for instance, a standard collecting semantics. 
Their use of two domains, concrete and abstract, with a Galois connection between them, for performing static analyses \textit{within a single domain} should not be confused with our idea of secure connections between independently-defined security lattices of two organisations.

At the systems level, there has been quite some work on SIF in a distributed setting. 
An exemplar is DStar \cite{zeldovich2008-nsdi}, which uses  sets of opaque identifiers to define security classes.
The DStar framework takes a \textit{particular} Decentralized Information Flow Control (DIFC) model \cite{Krohn2007-aa,zeldovich2006-osdi} for operating systems and extends it to a distributed network.
Subset inclusion is the (only) partial order considered in DStar's security lattice. 
Thus it is not clear if DStar can work on general IFC mechanisms such as FlowCaml \cite{Pottier2003-FlowCaml}, which can employ any partial ordering.
Nor can the DStar model express the labels of  JiF \cite{myers1999jflow} or Fabric \cite{liu2017fabric} completely.
DStar allows bidirectional communication between processes $R$ and $S$ only if $L_R \sqsubseteq_{O_R} L_S$ and $L_S \sqsubseteq_{O_S} L_R$, \textit{i.e.}, when there is an order-isomorphism between the labels. 
We have argued that such a requirement is far too restrictive for most practical arrangements for data sharing between organisations.

Fabric \cite{liu2009fabric,liu2017fabric} adds \textit{trust relationships} directly derived from a principals hierarchy to support systems with mutually distrustful nodes and allows dynamic delegation of authority.
It is not immediately clear whether that framework supports modular decomposition and analysis, a topic for future investigation.
Most of the previous DIFC mechanisms \cite{myers1999jflow, zeldovich2006-osdi, Krohn2007-aa, efstathopoulos2005asbestos, roy2009laminar, cheng2012aeolus} including Fabric are susceptible to the vulnerabilities mentioned in the motivating examples of our previous work \cite{BhardwajP2019}.

\section{Conclusions and Future Work}\label{sec:conclusion}

Our work follows Denning's proposal of lattices as the mathematical basis for analysis about secure information flows.   We segue to order-preserving morphisms between lattices as a natural framework for a scalable and modular analysis for secure inter-domain flows.
From the basic secure flow requirements that preserved the autonomy of the individual organisations, we identified the simple and elegant theory of Lagois connections as an appropriate formulation.  
Lagois connections provide us a way to connect the security lattices of two (secure) systems in a manner that does not expose their entire internal structure and allows us to reason only in terms of the interfaced security classes.

We have also illustrated that the theory of Lagois connections provides a versatile framework for supporting the discovery, decomposition, update and maintenance of secure MoUs for exchanging information between administrative domains.
Compositionality of Lagois connections provides the necessary modularity when chaining connections across several domains, while the canonical decomposition results provide methodological rules within which we can re-establish secure connections when the security lattices are updated. 

Moreover, as illustrated here, 
we have shown this framework is also applicable in more intricate information flow control formulations such as decentralised IFC and models with declassification \cite{myers-phd-tr-award}.  
Ongoing work indicates that the framework works smoothly in formulations with data-dependent security classes \cite{Lourenco2015-ug} as well.

Note that the secure Lagois connection between two domains, especially in the decentralised model, introduces new flows from principals in one domain to those in another, and conversely.
It can be argued that the \textit{sound} and \textit{complete} relabelling rule in the Decentralised Label Model
\[
P \vdash L_1 \sqsubseteq L_2 \mbox{~iff~}
(\forall P' \supseteq P) \textbf{X}(L_1, P') \supseteq \textbf{X}(L_2, P')
\]
(where $\textbf{X}(L, P')$ denotes the set of flows permitted by label $L$ given principals hierarchy $P'$), in a sense already accounts for these new flows.
From the viewpoint of one domain, the permitted flows to and within the other domain can be viewed as an extension $P'$ of its principals hierarchy $P$ that now incorporates principals from the other domain.
The Lagois conditions, however, guarantee that when data flow to another domain and back, \textit{no new flows} are created within each individual domain. 
Thus, our connections-based framework provides a modular approach to the static analysis of permitted flows, by partitioning the analysis to flows within each domain and inter-domain flows.
Indeed, the analysis is confined to the \textit{syntactic framework} of the principals hierarchies, and the associated policies and labels.
The proofs of correctness with respect to the semantics do not have to be reworked to consider the slew of new flows.

We believe that it is important to have a framework in which secure flows should be treated in a modular and autonomous manner for the following reason.
The notion of  a principal delegating to others the capacity to act on its behalf
(\textit{e.g.}, in the DIFC model of Myers \cite{myers-phd-tr-award}) does not scale well to large, networked systems since a principal may repose different levels of trust in principals on various hosts in the network. 
For this reason, we believe that frameworks such as Fabric \cite{liu2009fabric, liu2017fabric} may provide more power than mandated by a principle of least privilege.
In general, since a principal rarely vests unqualified trust in another in all contexts and situations, one should confine the influence of the principals possessing delegated authority to only specific domains.
A mathematical framework that can deal with localising trust and delegation of authority in different domains and controlling the manner in which information flow can be secured deserves a deeper study.  
We believe that mathematical theories such as Lagois connections provide the necessary structure for articulating these concepts.  

We conclude by noting that it is surprising that Lagois connections have not seen greater use in computer science and particularly in static analysis.  
Most applications of Galois connections in fact employ \textit{Galois insertions}, which also happen to be special cases of Lagois connections \cite{MELTON1994lagoisconnections}.
While the duality between closure and interior operators in Galois connections provides them an elegance, the quite different symmetries exhibited by Lagois connections also seem to be natural and useful in many settings. 

\bibliography{mybibfile}

\end{document}